\documentclass[11pt]{article}
\usepackage[numbers]{natbib}
\usepackage{macros}
\usepackage{formatting}

\usepackage[utf8]{inputenc} 
\usepackage[T1]{fontenc}    
\usepackage{hyperref}       
\usepackage{url}            
\usepackage{booktabs}       
\usepackage{amsfonts}       
\usepackage{nicefrac}       
\usepackage{microtype}      
\usepackage{xcolor}         
\usepackage{wrapfig}
\usepackage{subcaption}
\usepackage{tikz}
\usetikzlibrary{%
  calc,
  fit,
  arrows,
  arrows.meta,
  positioning,
  decorations.pathreplacing,
  decorations.shapes,
}

\begin{document}

\begin{center}
  {\huge Differences-in-Neighbors for Network Interference in Experiments} \\
  \vspace{.5cm} {\Large Tianyi Peng ~~~ Naimeng Ye ~~~ Andrew Zheng}\\
  \vspace{.2cm}
  {\large Columbia University ~~~ Columbia University ~~ University of British Columbia}\\
  \vspace{.2cm}
  \texttt{tianyi.peng@columbia.edu, naimeng.ye@columbia.edu, andrew.zheng@sauder.ubc.ca}
  
\end{center}

\begin{abstract}
Experiments in online platforms frequently suffer from network interference, in which a treatment applied to a given unit affects outcomes for other units connected via the platform. This SUTVA violation biases naive approaches to experiment design and estimation. A common solution is to reduce interference by clustering connected units, and randomizing treatments at the cluster level, typically followed by estimation using one of two extremes: either a simple difference-in-means (DM) estimator, which ignores remaining interference; or an unbiased Horvitz-Thompson (HT) estimator, which eliminates interference at great cost in variance. Even combined with clustered designs, this presents a limited set of achievable bias variance tradeoffs. We propose a new estimator, dubbed Differences-in-Neighbors (DN), designed explicitly to mitigate network interference. Compared to DM estimators, DN achieves bias second order in the magnitude of the interference effect, while its variance is exponentially smaller than that of HT estimators. When combined with clustered designs, DN offers improved bias-variance tradeoffs not achievable by existing approaches. Empirical evaluations on a large-scale social network and a city-level ride-sharing simulator demonstrate the superior performance of DN in experiments at practical scale.
\end{abstract}


\section{Introduction}

Experimentation is a ubiquitous learning method in online platforms, where the experimenter's goal is commonly to estimate a global {\it average treatment effect} (ATE): i.e., the impact of applying a given intervention to an entire population of experimental ``units'' (``global treatment''), as compared with the absence of that intervention for the entire population (``global control''). In typical A/B testing scenarios, one infers the ATE by randomly selecting experimental units to receive the intervention, and those to use as a control group, and naively differencing the outcomes in each group.

In complex systems with many interacting participants, however, such an approach suffers from {\it interference}: a violation of the Stable Unit Treatment Value Assumption (SUTVA) which arises when treatments applied to one unit impact outcomes at other units. Common examples include experiments in social networks, where connected users impact each other's behavior via communication; or, in marketplaces, where buyers impact each other by consuming available inventory. The pattern of interference between units can be generically modeled as a directed graph, in which experimental units are nodes, and an edge ($u, v$) exists if the treatment at $u$ has the potential to impact the outcome at $v$.

The most common approach to mitigating the bias introduced by interference is the {\it design} of clustered experiments, in which units that interfere with each other are randomized together, so as to better approximate conditions under global control or global treatment. Examples including clustering based on geography, time (as in a switchback experiment), or otherwise known network structure (as in social network experiments). The resulting clusters are usually assumed to be independent, despite potential interference occuring at cluster boundaries, and estimation proceeds as in a typical A/B test. However, while aggregating units into clusters can reduce bias, doing so incurs a cost in increased variance. Larger clusters imply fewer independent samples -- at one extreme, treating the entire population together as a single cluster enables unbiased estimation of the ATE, but error due to variance is on the order of the ATE itself. Ultimately, the feasible set of cluster designs is dictated by the tradeoff between bias and variance induced by the choice of cluster size.

Faced with these limits to clustered experiment {\it design}, a promising route to expanding the range of bias and variance trade-offs possible is through improved {\it estimation}. Beyond naive estimators, which simply ignore interference, existing work is focused on {\it unbiased} estimators, which (on average) totally eliminate the effects of interference. The two classes of approaches here are importance sampling (Horvitz-Thompson) estimators, which eliminate bias at a large cost in variance; and regression-type approaches, which apply for known, parametric outcome models, but which have no guarantees under model misspecification. Unfortunately, the former has impractically high variance; the latter requires unrealistic assumptions; and as a result both are difficult to apply in practice. The question of how to construct estimators which are simultaneously low bias, low variance, and applicable under general outcome models, remains open.

\subsection{Contributions}

Motivated by this gap, we introduce a novel estimator for the ATE under general network interference, which we dub Differences-In-Neighbors (DN). In short, for a large class of problems, the DN estimator has provably small bias relative to naive estimation, while simultaneously achieving variance exponentially smaller than unbiased Horvitz-Thompson (HT) estimators. Applied in combination with clustered designs, DN de-biases remaining interference between clusters, which enables the use of substantially smaller clusters and improved bias-variance tradeoffs as compared with naive estimation.

Below, we describe these contributions in greater detail:

{\bf 1. Second-order Bias}:  For a general class of outcome functions in which direct effects of interference dominate and are $O(\delta)$, the DN estimator achieves bias that is $O(\delta^{2})$; i.e., second order in the magnitude of interference. In general, the bias of DN can be bounded by a natural notion of ``smoothness'' of the outcome function. We derive the DN estimator as well as its bias bound from a Taylor expansion of the treatment effect, which also immediately yields a series of high-order DN estimators with different bias-variance trade-offs.


{\bf 2. Variance}: Compared with unbiased HT estimators, which have variance scaling exponentially in the maximum degree of nodes in the network, the variance of DN scales only polynomially with this degree. While this entails a substantial increase in variance over the naive difference-in-means (DM) estimator, we show that at practical scale the bias-variance trade-off afforded by DN results in estimators with substantially lower error than either alternative.

{\bf 3. Synergies with clustering} We also introduce a version of the DN estimator tailored to clustered experimental designs. We show, both theoretically and empirically, that this estimator achieves trade-offs which are impossible to attain simply through improved clustering designs.

{\bf 4. Practical performance} We demonstrate the effectiveness of the DN estimator on a range of large-scale network interference problems, including experiments in social network graphs and also as applied to a realistic ride-sharing simulator with spatio-temporal clustering. 


\subsection{Related literature}

{\bf Clustered experiment design.} There is a long line of work on clustered experiment designs combined with either DM or HT estimation. \cite{uganderGraphClusterRandomization2013} in particular shows that graphs whose scaling follows a so-called $\kappa$-restricted growth condition admit a clustering, which, when combined with HT estimation, enables estimation of the ATE at an $O(\frac{1}{\sqrt{N}})$ rate where $N$ is the number of nodes. This rate, however, scales exponentially in $\kappa$. To mitigate this dependence, \cite{UganderYin+2023} proposes randomizing the clustering itself, reducing the dependence to a polynomial scale. This comes at the cost of increased computational complexity due to the need for more intricate probability calculations in the HT estimator. \cite{eckles2017design} undertakes an empirical exploration of the interaction between cluster design and the choice between Naive and a variant of HT estimation. 
\cite{viviano2023causal} addresses the tradeoff we discuss most directly, choosing clusters to optimize the bias-variance tradeoff under a worst-case outcome model, assuming a simple DM estimator. Compared to the DM or HT estimators used in clustered experiment designs, the DN estimator offers a favorable bias-variance tradeoff curve, as demonstrated in  \cref{sec:small-world} and \cref{sec:experiments}.


\textbf{Outcome modeling.} Another line of work assumes a structured model of interference and designs regression-based approaches to mitigate the bias. For example, \cite{Leung2021RateoptimalCD, leung2024clusterrandomizedtrialscrossclusterinterference} assume a spatial interference model and propose methods based on radius truncation. In \cite{gui2015network}, a linear model is assumed for the interference effect, depending on the treated fraction of neighbors, and regression is used to estimate or remove the interference. \cite{Aronow_2017} introduces the more general concept of an "exposure mapping": a function that maps a vector \( z \in \{0, 1\}^{N} \) to an outcome. In full generality, such a mapping can describe any dependence between experimental units, but various specific exposure mappings are described for which estimation is more tractable.

In contrast to this body of literature, our DN estimator does not rely on a structural model assumption for the outcome function. Instead, its bias depends on the ``smoothness'' of the outcome function, which quantifies the extent to which outcomes under random treatment assignments provide information about outcomes under global treatment or global control. We will make this intuition precise in \cref{sec:prob-formulation}.

A recent line of work \cite{Cortez_Rodriguez_2023, eichhornLoworderOutcomesClustered2024} introduces a "low-order" outcome function for interference and proposes unbiased regression-based estimators, referred to as pseudo-inverse estimators. Interestingly, under a related outcome model, which is low-order, but where the model coefficients can depend on the affected node's treatment, DN coincides with the corresponding pseudo-inverse estimator. The analysis of DN, however, takes a very different approach/motivation: it holds regardless of whether the outcome function has this low-order form, and the derivation is based on a Taylor expansion for deriving a second-order (or higher-order) bias. This complements this line of work by providing an analysis of bias under misspecification, and may be of independent interest.



\textbf{Industrial practice for addressing interference.} Tackling network interference on online platforms is a long standing problem. \cite{xu2015infrastructure} provides an overview of LinkedIn's experimentation platform, explicitly identifying interference as a key challenge, and a well-known LinkedIn study \cite{gui2015network} proposes both clustering and estimation approaches to mitigate network interference.  Similarly, \cite{karrer2021network} describes Facebook’s network clustering experiments, highlighting practical considerations in large-scale deployment. Beyond social networks, \cite{Holtz2020LimitingBF} examines interference in large-scale online marketplaces, using simulations to evaluate experimental methods in settings where interference is difficult to specify precisely. The persisting challenges of conducting clustered experiments in practice is the motivation of this work. 

{\bf Differences-in-Qs estimation.} Finally, our methodology is inspired by recent work addressing a form of "Markovian" temporal interference commonly encountered in dynamical systems, utilizing Difference-in-Qs (DQ) estimators \cite{fariasMarkovianInterferenceExperiments2022,fariasCorrectingInterferenceExperiments2023}. The development of the DN estimator arose from the following insight: \textit{that Q-values in the MDP system can be interpreted as the sum of all (outconnected) neighbors in a general graph}. This insight suggests that the underlying principle of DQ estimators can be extended beyond their original scope, applying not only to temporal graphs but also to spatio-temporal systems and more general interference networks. We will formalize this connection in \cref{sec:outcome-model}.

\begin{figure}
    \centering
\includegraphics[width=\linewidth]{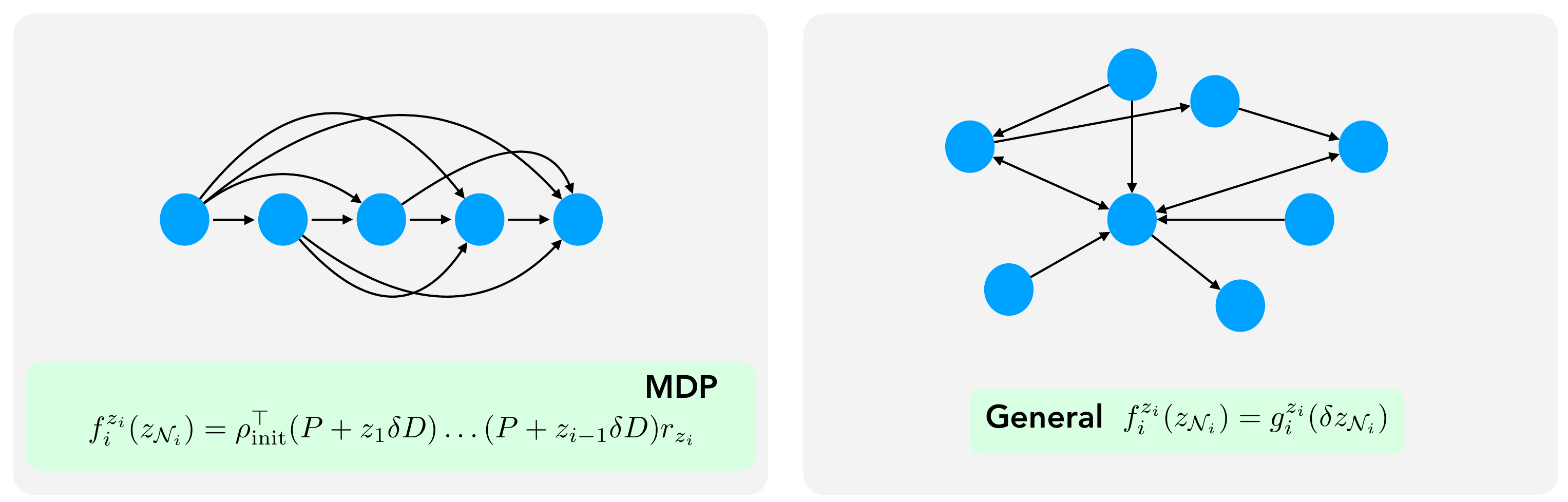}
    \caption{(Left) Temporal interference, where interference is captured through a Markov chain, is often considered a distinct class of problems from (Right) network interference, where interference is captured by interactions between nodes. On view of our work is a generalization of the Difference-in-Q estimators \cite{fariasMarkovianInterferenceExperiments2022,fariasCorrectingInterferenceExperiments2023}, which were proposed to address temporal interference in Markov chains, to network interference settings. This is achieved by observing that the Q-value is essentially the sum of (out-connected) neighbors in a general graph. This generalization enables a unifying view of temporal and network interference, allowing techniques developed for one domain to potentially benefit the other. The connection is formalized in \cref{sec:outcome-model}.}
    \label{fig:markov}
\end{figure}


Beyond DQ estimation, there is a growing body of work that aims to quantify and reduce bias in marketplace interference either through design or estimation such as \cite{Li_twosided_market,bajariMultipleRandomizationDesigns2021,johariExperimentalDesignTwoSided2020,weng2024experimental,qu2021efficient,xiong2024data,xiong2024optimal,huang2023estimating,ni2023design,bojinov2023design,cortez2024combining,jia2023clustered,shirani2024causal,zhao2024simple} and many more; our work contributes to this trend.

\section{Problem Formulation}
\label{sec:prob-formulation}
\subsection{Experimentation under network interference}
We consider the problem of experimenting on a population of $N$ units under potential network interference. We model the interference between units in this population as a undirected graph $G = ([N], E)$, where we have that an edge $(i, j)$ exists in $E$ if a treatment applied to unit $i$ can impact outcomes for unit $j$. As such, we will be interested in the neighbors of any given node, where neighbors of $i$ are denoted as $\mathcal{N}_{i}  =\{ j \in [N]: (j, i) \in E \}.$ 
We will find it useful to treat direct treatment effects separately from interference effects, and as such we will assume no self loops: i.e., $i \notin \mathcal{N}_{i}$. Note also that all the estimators and results in this paper generalize immediately to directed interference graphs; however, for simplicity of exposition, we will restrict our attention here to undirected graphs. Several bounds we discuss will depend on the maximum degree of the graph, which we denote as $d = \max_{i} |\mathcal{N}_i|.$

When running an experiment, we assign a treatment $z_{i} \in \{0, 1\}$ to each unit $i$. We denote the vector of treatment assignments across all nodes simply as $z \in \{0, 1\}^{N}$. Define for each node a potential outcomes function $f_{i} : \{0, 1\}^{N} \mapsto \mathbb{R}$, where in full generality the outcomes for node $i$ may depend on the entire vector of treatments. Throughout this work, we will specialize this slightly via an assumption of {\it neighborhood interference}, in which the outcomes at any node $i$ depend only on its own treatment $z_{i}$, and the treatments of its neighbors $\{z_{j}: j \in \mathcal{N}_{i}\}$:

\begin{assumption}[Neighborhood Interference] 
\label{ass:neighborhood}
For all $i \in [N]$, and treatment vectors $z, z'$, if $z_{j} = z'_{j}$ for all $j \in \{i\} \cup \mathcal{N}_{i}$, then $f_{i}(z) = f_{i}(z')$.
\end{assumption}

 We can now state our problem. An experimenter designs a treatment assignment policy $\pi$, which defines a distribution over treatment vectors. The experimenter draws a random treatment vector $z \sim \pi$, and subsequently observes the outcomes $Y_1 \ldots Y_{N}$ where $Y_i = f_{i}(z)$. Based on these observations, the experimenter's goal is to estimate the {\it average treatment effect} (ATE), i.e., the difference in outcomes between global treatment and global control, averaged over all units. This is defined as:

  \begin{equation*}
    {\rm ATE} = \frac{1}{N} \sum_{i=1}^{N} (f_{i}({\vec 1}) - f_{i}({\vec 0}))
  \end{equation*}
  where {$\vec 1$}, {$\vec 0$} respectively denote the vectors of all ones and all zeros.

\subsection{Outcome model}
\label{sec:outcome-model}


In this work, we will study the bias and variance of various estimators under network interference. As such, it will be useful to parameterize our outcome functions to allow a concise description of the {\it strength} of the interference. Here we describe a generic class of outcome functions which admit such a description.

First, in many real-world experiments, the outcome $Y_i$ at node $i$ can be expected to depend most heavily on the direct treatment $z_i$, which can both induce a direct treatment effect as well as mediate the effect of interference. To make this dependence clear, without loss of generality, we define the conditional outcome functions $f^0(z_{\mathcal{N}_i})$ and $f^1(z_{\mathcal{N}_i})$ such that 

\begin{equation}
\label{eqn:outcome}
  f_{i}(z) = \begin{cases}
  f_{i}^{0}( z_{\mathcal{N}_{i}}) &\text{~if~}  z_{i} = 0 \\ 
  f_{i}^{1}( z_{\mathcal{N}_{i}})&\text{~if~} z_{i} =1 \end{cases}
\end{equation}


Second, we will assume that each conditional outcome function $f_i^{z_i}$ can be re-parameterized as $f_i^{z_i}(z_{\mathcal{N}_i}) = g^{z_i}_{i}(\delta z_{\mathcal{N}_i})$, for some function $g^{z_i}_i : [0, 1]^n\mapsto \mathbb{R}$, and $\delta > 0$. Here $\delta$ is a parameter that scales the magnitude of the interference that any individual neighbor has on the outcome $Y_{i}$. Again, this suffers no loss of generality. However, in our subsequent analysis, we will be interested in how the bias of various estimators scales in terms of the interference strength $\delta$. As such, we will be most interested in outcome functions for which scaling $\delta$ has a natural interpretation; we provide several examples at the end of this section.



Without further assumptions on the potential outcomes $f_{i}$, each of the $2^{\mathcal|{\mN}_{i}|}$ possible treatment configurations can still produce arbitrary outcomes, making it impossible to draw useful conclusions about the treatment effect $Y_{i}({\vec 1}) - Y_{i}({\vec 0})$ at node $i$ without actually observing outcomes under $z_{\mathcal{N}_{i}} = {\vec 1}$ or $z_{\mathcal{N}_{i}}= {\vec 0}$. To enable more practical inference, we will need to make an assumption on the model: we will restrict our attention to outcome models where the individual outcome functions $f^{1}_{i}$ and $f^{0}_{i}$ vary {\it smoothly} in the strength of the interference. Precisely

\begin{assumption}[Smooth interference]
  \label{ass:smooth}
  For any node $i$, and neighbors $j, k \in \mathcal{N}_{i}$ where $j\neq k$, and for any treatment $z_{i} \in \{0, 1\}$, the second-order derivatives of $g_{i}^{z_{i}}(x)$ at any $x \in [0, \delta]^{N}$ are bounded as
  \begin{equation*}
    \left\vert \frac{\partial g_{i}^{z_{i}}}{\partial x_{j} \partial x_{k}} \right\vert \leq L.
  \end{equation*}
  \end{assumption}


  Intuitively, by preventing the outcome function from varying too abruptly in the treatments, this assumption ensures that any realization of the treatment vector $z$ provides a certain amount of information regarding outcomes under global treatment or global control.\footnote{A natural question is how assumptions on first-order or higher-order derivatives instead would affect our analysis. In fact they lead to different variants of DN with different order of estimation errors.  See Section~\ref{sec:proof-bias-main} for a  brief discussion for such generalizations.}  As it turns out, a wide range of outcome models in the literature do fall naturally into this framework, and come equipped with natural bounds on $L$:

  
\paragraph{Linear outcome models} This includes an array of common models consisting of additive effects from each treated neighbor, such as models where the outcome $Y_{i}$ depends on the number or the proportion of treated neighbors, for example, $f^{z_{i}}(z_{\mathcal{N}_{i}}) = \alpha + \beta z_{i} + \delta \gamma 1^\top z_{\mathcal{N}_{i}}$ or more generally
\begin{align*}
f^{z_{i}}(z_{\mathcal{N}_{i}}) = \alpha_i + \sum_{j \in \mN_i} g_{ij}(z_i, z_j).
\end{align*} for some arbitrary function $g_{ij}$. Such models clearly have no second-order terms, and therefore satisfy \cref{ass:smooth} with $L=0$.

\paragraph{Multiplicative outcome models} Consider the multiplicative effects model 
$$f^{z_{i}}(z_{\mathcal{N}_i}) = c_{0} \prod_{j \in \mathcal{N}_{i}} \left(1 + \frac{\delta}{|\mathcal{N}_i|} c_{ij} z_{j}\right).$$
Here, we have $L$ bounded as $O(\delta^2)$. 

\paragraph{Low-order outcome models} \cite{Cortez_Rodriguez_2023} proposes a general class of outcome functions, where interference effects depend on interactions between up to $\beta$ of a node's neighbors. Such models can be expressed as degree-$\beta$ polynomial models of the form $f^{z_{i}}( z_{\mathcal{N}_{i}}) = c_{0} + \sum_{k=1}^{\beta} \sum_{S \in \mathcal{S}_{i}^{k}} c_{k, \mathcal{S}} \prod_{j \in S} \delta z_{j}$,
where $\mathcal{S}_{i}^{k} = \{ S \subseteq \mathcal{N}_{i} : |S| \leq k \}$ is the set of size $k$ subsets of neighbors of $i$. Under this general model, we have that $L$ is bounded by the coefficients corresponding to higher-order interactions between treatments; i.e., $L \leq \sum_{k = 2}^{\beta} \delta^{k - 2} \sum_{S \in \mathcal{S}^{k}_{i}} c_{k, S}$.

\paragraph{Markovian interference} Following \cite{fariasMarkovianInterferenceExperiments2022,fariasCorrectingInterferenceExperiments2023}, consider experimentation in a discrete-time dynamical system where nodes $i$ represent time steps, and treatments $z_{i}$ represent actions. The interference graph then contains an edge from every time $i$ to every subsequent time $j > i$ (see Figure~\ref{fig:markov}). The system has $S$ possible states, intially distributed according to $\rho_{\rm init} \in \mathbb{R}^{S}$. When $z_{i} = 0$, states evolve according to the transition matrix $P \in \mathbb{R}^{S \times S}$, whereas applying the treatment $z_{i}=1$ modifies the transition probabilities to $P + \delta D$, for some $D \in \mathbb{R}^{S \times S}$, appropriately normalized so that $P + \delta D$ remains a stochastic matrix. Finally, conditional on the state at time $i$, which may depend on prior treatments $z_{j}$ for $j < i$, outcomes (i.e., rewards) only depend on the current treatment $z_{i}$. As such there are two reward functions $r_{1}, r_{0} \in \mathbb{R}^{S}$. Putting this together, we have the outcome function for $f^{z_{i}}_{i}( z_{\mathcal{N}_{i}}) = \rho^{\top}_{\rm init} \left[\prod_{j < i} (P + z_{j}  \delta D)\right] r_{z_{i}}$. \cite{fariasMarkovianInterferenceExperiments2022,fariasCorrectingInterferenceExperiments2023} provide explicit bounds for $L$ in terms of the ``effective horizon'' of the Markov chain.

\section{The Differences-In-Neighbors Estimator}\label{sec:DN-unit}

We begin by introducing our proposed estimators for the ATE, and their analysis, under a simple Bernoulli randomization design where $z_{i} \overset{\rm iid}{\sim} {\rm Bernoulli}(p)$.

\paragraph{The Naive Estimator} The most common approach in practice is simply to ignore any interference, and proceed with estimation as one would under SUTVA. A typical, biased approach here is the naive Difference-in-Means (DM) estimator: \footnote{Empirically, ${\rm \tilde{ATE}}_{\rm DM} = \frac{1}{\sum_{i=1}^{N} z_i}\sum_{i=1}^{N} z_iY_{i} - \frac{1}{\sum_{i=1}^{N} 1-z_i}\sum_{i=1}^{N} (1-z_i) Y_{i}$ is often used. Here we present its alternative form for simplifying the analysis. $\rm \tilde{ATE}_{\rm DM}$ is used in our experiments for comparison.}

  \begin{equation}
    \label{eq:naive}
    {\rm ATE}_{\rm DM}
    = \frac{1}{N}\sum_{i=1}^{N} \left(\frac{z_i}{p} -  \frac{1-z_i}{1-p}\right) Y_i
  \end{equation}


\paragraph{The Differences-In-Neighbors Estimator} Intuitively, the bias of the DM estimator arises from the fact that it makes no effort to account for the network effects of applying a treatment $z_i$. In particular, \eqref{eq:naive} only ``attributes'' a node's own outcome $Y_{i}$ to its treatment. We construct the Differences-In-Neighbors (DN) estimator by making a simple and intuitive correction: when measuring the impact of a treatment $z_{i}$, one should also include the outcomes of the neighbors  of $i$ impacted by that treatment. More precisely, we propose the following estimator:

\begin{align}\label{eq:DN-p-probability}
    {\rm ATE}_{\rm DN} = \frac{1}{N}\sum_{i=1}^{N} \left(\frac{z_i}{p} -  \frac{1-z_i}{1-p}\right) \left(\sum_{j \in \mN_i}  Y_{j} \cdot \xi_{j} + Y_i\right)
\end{align}
where $\xi_{j} = \frac{z_j(1-p)}{p} + \frac{(1-z_j)p}{1-p}$ is a propensity score correction. When $p=1/2$ for all nodes, we have $\xi_{j}=1$ and this simplifies to an intuitive form:
\begin{align}\label{eq:DN-equal-probability}
    {\rm ATE}_{\rm DN} = \frac{1  }{N}\sum_{i=1}^{N} \left(\frac{z_i}{p} -  \frac{1-z_i}{1-p}\right) \left(\sum_{j \in \mN_i \cup \{i\}} Y_{j}\right)
\end{align}

That is, we use the difference between the aggregated outcomes of node $i$ and its neighbors; hence the name \textit{Differences-in-Neighbors}.

\paragraph{Credit Assignment.} Whereas equation \cref{eq:DN-p-probability} takes the perspective of the treatment $z_{i}$, and measures its impact on $i$'s neighbors, we can equivalently write the DN estimator as ``assigning credit'' for the outcome $Y_{i}$ to each of the treatments $\{ z_{j}: j \in \mathcal{N}_{i}\}$ that contributed to that outcome. This alternative view will be useful in the subsequent analysis:

\begin{align}\label{eq:DN-assignment-view}
    {\rm ATE}_{\rm DN} = \frac{1}{N} \sum_{i=1}^{N}  \left(\left(\frac{z_i}{p} - \frac{1-z_i}{1-p}\right) +\xi_{i}\sum_{j \in \mN_{i}}\left(\frac{z_j}{p} - \frac{1-z_j}{1-p}\right)\right)Y_{i}
\end{align}

\subsection{Summary of Guarantees}

The central contribution of this work is to show that the DN estimator achieves a favorable bias-variance trade-off, which is not possible under existing estimators. In this section, we state these claims with detailed explanations and proofs deferred.
\paragraph{The Horvitz-Thompson Estimator}
To begin, we need to introduce one more common estimator: the Horvitz-Thompson (HT) estimator, which eliminates interference bias via importance sampling:

\[
{\rm ATE_{\rm HT}} = \frac{1}{N}\sum_{i=1}^{N} \left(\prod_{j \in \mathcal{N}_i}\frac{z_{j}}{p} -  \prod_{j \in \mathcal{N}_i} \frac{1 - z_{j}}{1 - p} \right) Y_i.
\]

Table~\ref{tab:results-summary} summarizes bounds on bias and variance of each estimator. The naive DM and HT estimators represent opposite extremes of the bias-variance spectrum; whereas DN simultaneously achieves second-order bias and variance scaling only polynomially in the degree $d$. In our numerical experiments later, we will show that in practical regimes, this trade-off allows DN to achieve substantially lower RMSE than existing alternatives.

    \begin{table}[htbp]
      \centering
      \begin{tabular}{c|cc}
        Estimator & Bias & Variance \\
        \hline 
        DM & $\Theta(d \delta)$ & $\Theta(\frac{1}{N})$\\
        DN & $\Theta(d^2\delta^{2})$ & $O(\frac{d^{4}}{N})$ \\
        HT & 0 & $\Theta(\frac{2^{d}}{N})$
      \end{tabular}
      \caption{Summary of bias and variance results for DM, DN, and HT estimators. DN simultaneously achieves bias second order in the strength of interference, while paying a variance cost polynomial in the degree $d$.}
      \label{tab:results-summary}
    \end{table}

\subsection{Bias}\label{sec:proof-bias-main}

In this section, we bound the bias of the DN estimator in estimating the ATE. We also present the proof of this bound, which provides a more precise intuition for the derivation of the DN estimator as a first-order Taylor approximation to the ATE. 

We first state a general version of the bound which applies to {\it any} outcome function, not just those posited in \cref{sec:outcome-model}; we will specialize it to our outcome model in the sequel.

We will state this bound in terms of a {\it discrete} notion of the ``smoothness'' of the outcome function, defined as follows. First, for any $z \in \{0,1\}^{n}$, for any $j, k \in[n]$, let $z^{(z_j=a,z_k=b)}$ denote $z$ with its $j^{\rm th}$ and $k^{\rm th}$ elements replaced by $a, b$ respectively. Then, for a function $h : \{0, 1\}^{n} \mapsto \mathbb{R}$, we define the second-order finite differences as

\begin{equation*}
\Delta_{jk}h(z) = h(z^{(z_{j}=1, z_{k}=1)})
- h(z^{(z_{j}=0, z_{k}=1)})
- h(z^{(z_{j}=1, z_{k}=0)})
+ h(z^{(z_{j}=0, z_{k}=0)})
\end{equation*}

\textit{Finally, we say that $h$ is $\epsilon$-smooth if for any $z\in\{0,1\}^{N}$ and $j, k \in [N]$}, 
\begin{align*}
    |\Delta_{jk} h(z) | \leq \epsilon.
\end{align*}

We can now state our general bound:

\begin{theorem}
    \label{th:general-bias}
    Suppose that $f^{1}_{i}$ and $f^{0}_{i}$ are $\epsilon$-smooth for all nodes $i \in [N]$. Then,
    \begin{align*}
        | {\rm ATE} - {\sf{E}}[ {\rm ATE}_{\rm DN}]| \leq d^{2} \epsilon
    \end{align*}
\end{theorem}

For outcome models of the form specified in  \cref{sec:outcome-model}, we can relate this bound directly to the magnitude of interference $\delta$ via the following Corollary:

\begin{corollary}
    \label{dn-bias}
    For an outcome model satisfying \cref{ass:smooth}, ${\rm ATE}_{\rm DN}$ has bias bounded as 
    \begin{align*}
        | {\rm ATE} - {\sf{E}}[ {\rm ATE}_{\rm DN}]| \leq d^2\delta^2 L.
    \end{align*}
\end{corollary}


\begin{proof}
The statement follows by bounding the second-order finite differences using \cref{ass:smooth}. For any $j, k \in \mathcal{N}_i$, $a \in \{0, 1\}$, and $z \in \{0, 1\}^{|\mathcal{N}_i|}$, by the fundamental theorem of calculus: 

\begin{align}
\label{assum:second-order-finite-difference}
\Delta_{jk} f^a_i(z)  &= |f_i^a(z^{(z_j=1, z_{k}=1)})-f^a_i(z^{(z_j=1, z_{k}=0)})-f^a_i(z^{(z_j=0, z_{k}=1)}) +f^a_i(z^{(z_j=0, z_{k}=0)})| \nonumber\\
    &= \left|\int_{z_{j}=0}^{1}\int_{z_{k}=0}^{1} \frac{\partial f^a_{i}}{\partial z_{j} \partial z_{k}} d_{z_{j}}d_{z_{k}}\right|\nonumber\\
    &\overset{(i)}{=} \left|\int_{x_{j}=0}^{\delta}\int_{x_{k}=0}^{\delta} \frac{\partial g_{i}^a}{\partial x_{j} \partial x_{k}} d_{x_{j}}d_{x_{k}}\right|\nonumber\\
    &\leq \int_{x_{j}=0}^{\delta}\int_{x_{k}=0}^{\delta} \left|\frac{\partial g_{i}^{a}}{\partial x_{j} \partial x_{k}} \right| d_{x_{j}}d_{x_{k}}\nonumber\\
    &\overset{(ii)}{\leq} L\delta^2
\end{align}
where (i)  uses a change of variables based on the parameterization $f_i^a(z) = g(\delta z)$, and (ii) uses \cref{ass:smooth}.
\end{proof}

\subsection{Proof of Theorem~\ref{th:general-bias}}
  We begin by analyzing the expected value of ${\rm ATE}_{\rm DN}$. Rather than computing this quantity directly, it will be instructive to instead frame ${\rm ATE}_{\rm DN}$ as an unbiased estimator for an intuitive quantity: a first-order Taylor approximation to the ATE.
  
  Here, we analyze a single node $i$, and denote by $f(z)$ the outcomes for node $i$, omitting the subscript $i$ for simplicity. Let $Z \in \{0, 1\}^{|\mathcal{N}_{i}|}$ be a random vector of treatment assignments where $Z_{j} \overset{\rm iid}{\sim}. {\rm Bernoulli}(p)$. Let $p = \{ p : j \in \mathcal{N}_{i}\}$ denote the vector of treatment probabilities, and define the expected outcomes $F^{1}(w) = {\sf E}[f^{1}(Z) ]$ and $F^{0}(w) = {\sf E}[f^{0}(Z) ]$ where expectations are taken over the treatment assignments.

  Using this notation, we can write the contribution of node $i$ to the ATE simply as $F^{1}(1) - F^{0}(0)$. Under the experiment that we propose, a given realization of $Z$ can be used to construct an unbiased estimate of either $F^{1}(p)$ or $F^{0}(p)$; our goal will now be to use such an estimate to approximate $F^{1}(1)$ and $F^{0}(0)$ respectively. We begin by approximating $F^{1}(1)$, which has the following explicit expression in terms of possible realizations of $Z$:

\begin{equation*}
  F^{1}({w}) = \sum_{{z} \in \{0, 1\}^{|\mathcal{N}_{i}|}} {\sf Pr}(Z = {z}) f^{1}({z}) = \sum_{z \in \{0, 1\}^{|\mathcal{N}_{i}|}} \left(\prod_{j \in \mathcal{N}_{i}} (z_{j} w_{j}  + (1 - z_{j}) (1 - w_{j})) \right) f^{1}(z)
\end{equation*}

Then we can differentiate $F^{1}$ with respect to the treatment probabilities:

\begin{align*}
\frac{\partial F^{1}}{\partial w_{j}} &= \sum_{z \in \{0, 1\}^{|\mathcal{N}_{i}|}} (z_{j} - (1 - z_{j})) \left(\prod_{k \in \mathcal{N}_{i} \setminus \{j\}} (z_{k} w_{k}  + (1 - z_{k}) (1 - w_{k})) \right) f(z)\\
&= {\sf E}[f^{1}(Z) | Z_j = 1] - {\sf E}[f^{1}(Z) |  Z_j = 0]
\end{align*}

Now, we construct a first-order Taylor approximation for $F^1(1)$ around $F^1({p})$. This gives:
\begin{align*}
    F^1({p}) + (1-{p})^\top \nabla F^1(p)
    &= F^1({p}) + \sum_{j\in \mN_i} (1-p)\frac{\partial F^{1}}{\partial w_{j}} \Big\vert_{w=p}\\
    &= F^1({p}) + \sum_{j\in \mN_i} (1-p)\left({\sf E}[f^{1}(Z) |  Z_j = 1] - {\sf E}[f^{1}(Z) | Z_j = 0]\right)
\end{align*}

Similarly, we can approximate $F^{0}(0)$ as:
\begin{align*}
    F^0({p}) + (1-{p})^{\top} \nabla F^0(p) &= F^0({p}) + \sum_{j\in \mN_i} (0-p)\left({\sf E}[f^{0}(Z) | Z_j = 1] - \E[f^{0}(Z) | Z_j = 0]\right)
\end{align*}

Taken together, we have a first-order approximation of node $i$'s contribution to the ATE:

\begin{align*}
  F^{1}(1) - F^{0}(0) 
  &\approx
  F^1({p}) + \sum_{j\in \mN_i} (1-p)\left({\sf E}[f^{1}(Z) |  Z_j = 1] - {\sf E}[f^{1}(Z) | Z_j = 0]\right) \\
&\quad  - F^0({p}) - \sum_{j\in \mN_i} (0-p)\left({\sf E}[f^{0}(Z) | Z_j = 1] - {\sf E}[f^{0}(Z) | Z_j = 0]\right)
\end{align*}

It remains to construct an estimator for this idealized quantity. We will do so using propensity weights. Starting with the zeroth-order term, we note that

\begin{align*}
  F^{1}(p) &= {\sf E}[f^{1}(Z)] = {\sf E}\left[\frac{Z_{i}}{p} Y_{i}\right] &
  F^{0}(p) &= {\sf E}[f^{0}(Z)] = {\sf E}\left[\frac{1 - Z_{i}}{1-p} Y_{i}\right]
\end{align*}

Estimating each expectation with a single sample and taking the difference immediately yields the naive DM estimator; this demonstrates that {\it DM is a zeroth-order approximation to the ATE}.

The first-order terms can also be estimated via propensity weights; in particular, for $a, b \in \{0, 1\}$,

\begin{equation*}
{\sf E}[f^{a}(Z)| Z_{j} = b] = {\sf E}\left[\frac{\mathbf{1}\{Z_{i}= a\}}{{\sf Pr}(Z_{i} = a)}\frac{\mathbf{1}\{Z_{j}= b\}}{{\sf Pr}(Z_{j} = b)} Y_{i} \right]
\end{equation*}

Taking single-sample estimates of each expectation in the first-order approximation to the ATE, and rearranging these terms, then immediately yields the DN estimator. In summary: {\it the DN estimator is a unbiased estimate of a first order approximation to the ATE}.

To analyze the bias of the DN estimator, then, we need to bound the higher-order terms in the Taylor expansion.  For either treatment $a \in \{0, 1\}$ , we have


\begin{align*}
\frac{\partial^{2} F^{a}}{\partial w_{j} \partial w_{k}}
&= \sum_{z \in \{0, 1\}^{|\mathcal{N}_{i}|}} (z_{j} - (1 - z_{j})) (z_{k} - (1- z_{k})) \left(\prod_{l \in \mathcal{N}_i \setminus \{j,k\}} (z_{l} w_{l}  + (1 - z_{l}) (1 - w_{l}) \right) f(z) \\
&= {\sf E}\left[ f^{a}(Z) | Z_{j} = 1, Z_{k} = 1 \right] + {\sf E}\left[ f^{a}(Z) | Z_{j} = 0, Z_{k} = 0 \right] \\
&\quad - {\sf E}\left[ f^{a}(Z) | Z_{j} = 1, Z_{k} = 0 \right] - {\sf E}\left[ f^{a}(Z) | Z_{j} = 0, Z_{k} = 1 \right] \\
&= {\sf E}\left[ \Delta_{jk} f^a(Z)\right]
\end{align*}
for $j,k\in \mN_i$ and $j\neq k$, and $\frac{\partial^{2} F}{\partial w_{j}^{2}} = 0$ for $j = k$. 

Let $H^a(p)$ denote the matrix of second-order partial derivatives of $F^a$ at $p$. There exists some $p'$ such that the error of the DN estimator's first-order approximation to $F^a(a)$ is $(a - p)^\top H^a(p')(a - p)$. We conclude by bounding these terms, and the theorem follows immediately.
\begin{align*}
 \left\vert (a-{p})^\top H^a ({p}')(a-{p})\right\vert
 &= \sum_{j, k \in \mathcal{N}_i: j \neq k} (a - p)^2 {\sf E}[\Delta_{jk} f^a(Z)]  
 \leq d^2 \epsilon 
\end{align*}

    This also immediately shows how to derive higher-order versions of the DN estimator, along with their bias analyses, and also shows that for $n$ nodes the $n^{\rm th}$-order expansion is unbiased.

\subsection{Variance analysis}
We now turn to the variance of the DN estimator.
\begin{theorem}
    \label{dn-variance}
    \begin{align*}
        {\rm Var}({\rm ATE}_{\rm DN})\leq O\left(\frac{Y_{\max}^2}{N}\cdot \left(d^4+\frac{d^3}{p(1-p)}+ \frac{d}{p^2(1-p)^2}\right)\right)
    \end{align*}
\end{theorem}
\begin{proof}
We present the detailed proof in Appendix~\ref{dn-variance-proof}, and only sketch out a brief calculation here. For simplicity of notation, define the quantities: 

\[\eta_{i} = \frac{z_i}{p} - \frac{1 - z_i}{1-p}\qquad \xi_{i} = \frac{z_i(1-p)}{p} + \frac{(1-z_i)p}{1-p}\]
Simple calculation gives us:
\begin{equation}
\label{basic-property}
    \mathbb{E}[\eta] = 0
    \quad \mathbb{E}[|\eta|] = 2 \quad \mathbb{E}[\xi] = 1\quad \mathbb{E}(\eta^2) = \frac{1}{p(1-p)}\quad \mathbb{E}(\xi^2) = \frac{3p^2-3p+1}{p(1-p)} 
\end{equation}
We can expand out the variance and use bilinearity of the covariance to obtain the following:
    \begin{align}
    \label{eqn:dn-variance-expansion}
        {\rm Var}({\rm ATE}_{\rm DN}) &= {\rm Var}\left(\frac{1}{N}\sum_{i}^NY_i \left(\eta_i + \xi_i \sum_{j\in \mN_i}\eta_j\right)\right)\nonumber\\
        & = \frac{1}{N^2}\sum_{i}^N\sum_{j}^N{\rm cov}\left(Y_i\left(\eta_i+\xi_i\sum_{k\in N_i}\eta_k\right),Y_j\left(\eta_j+\xi_j\sum_{l\in N_j}\eta_l\right)\right)
    \end{align}
    Note that for any $i,j\in[N]$, if $i$ and $j$ does not share any neighbor nodes, clearly we have ${\rm cov}\left(Y_i\left(\eta_i+\xi_i\sum_{k\in N_i}\eta_k\right),Y_j\left(\eta_j+\xi_j\sum_{l\in N_j}\eta_l\right)\right) = 0$. Hence we can reduce $j$ to those that share at least one neighbor with $i$, denote this set as $\mathcal{M}_i$. Note that $|\mathcal{M}_i|\leq d^2$. Then Equation~\ref{eqn:dn-variance-expansion} can be further reduced to
    \begin{align*}
        &\sum_{i}^N\sum_{j\in \mathcal{M}_i}{\rm cov}\left(Y_i\left(\eta_i+\xi_i\sum_{k\in N_i}\eta_k\right),Y_j\left(\eta_j+\xi_j\sum_{l\in N_j}\eta_l\right)\right)\\
        &\leq Nd^2 {\rm cov}\left(Y_i\left(\eta_i+\xi_i\sum_{k\in N_i}\eta_k\right),Y_j\left(\eta_j+\xi_j\sum_{l\in N_j}\eta_l\right)\right)\\
        &\lessapprox O\left(Y_{\max}^2Nd^2\cdot d^2\E(|\xi_i||\eta_k||\xi_j|\eta_l| )+ Y_{\max}^2Nd^2 \cdot d \cdot \E(|\eta_u\xi_i\eta_k\xi_l|) \right)\\
        &\approx O(Y_{\max}^2 Nd^4 +Y_{\max}^2 Nd^3/p(1-p)).
    \end{align*}
\end{proof}

\section{Clustering with DN}
\label{sec:dn-cluster}
In this section, we extend the DN (Difference-in-Neighbors) estimator from unit-level randomization to cluster-level randomization designs, enabling a bias-variance tradeoff curve across varying cluster granularity. This allow us to surpass the limitations of traditional clustering approaches. To formalize this, let \(\mathcal{C} = \{C_{1}, C_{2}, \dotsc, C_{|\mathcal{C}|}\}\) denote a partition of the node set $[N]$. Under cluster-level randomization, each cluster \(C_{i}\) is assigned treatment with probability \(p\), independent of other clusters. Let \(z_{C_i}\) represent the treatment assignment indicator for cluster \(C_i\).

A straightforward approach is to contract the graph and treat each cluster as a single node, reducing the problem to unit-level randomization. However, we demonstrate in Section~\ref{sec:small-world} that this approach is suboptimal, since it fails to fully leverage the graph structure and unit-level observations.

To address this limitation, we propose a novel estimator that computes DN at the individual level while aggregating neighbor information at the cluster level. Specifically, for each node \(i\), let \(\mN_i^C\) denote the set of clusters containing at least one neighbor of node \(i\). The proposed estimator is defined as:

\begin{align}\label{eq:DN-clustering}
{\rm ATE}_{\rm DN-Cluster} = \frac{1}{N} \sum_{i=1}^{N}  \left(\left(\frac{z_i}{p} - \frac{1-z_i}{1-p}\right) +\xi_{i}\sum_{C_j \in \mN_i^C}\left(\frac{z_{C_j}}{p} - \frac{1-z_{C_j}}{1-p}\right)\right)Y_{i},
\end{align}
We can show that the estimator \({\rm ATE}_{\rm DN-Cluster}\) exhibits second-order bias with a reduced dependence on degree compared to the unit-level design due to the clustering structure. Formally, let $d_i$ be the degree of node $i$ and \(d_{i}^{C}\) denote the number of neighbors of node \(i\) that belong to the same cluster. We establish the following theorem:

\begin{theorem}
    \label{thm:dn-cluster-bias}
    The bias of the cluster-level DN estimator satisfies:
    \[
    | {\rm ATE} - \mathbb{E}[ {\rm ATE}_{\rm DN-Cluster}]| = O\left(\frac{\sum_{i=1}^{N}(d_i-d^{C}_i)^2}{N}\delta^2\right),
    \]
\end{theorem}

Notably, the bias in the cluster-level DN estimator scales quadratically not only with respect to \(\delta\) but also with \((d_i-d_i^{C})^2\), rather than \(d_i^2\). This highlights an additional dimension for bias reduction that is orthogonal to assumptions over the outcome function.

\begin{table}[htbp]
      \centering
      \begin{tabular}{c|cc}
        Cluster-Level Estimator & Bias & Variance \\
        \hline 
        DM & $\Theta\left(\frac{\sum_{i}\left(d_i - d^{C}_i\right)}{N} \delta\right)$ & $\Theta\left(\frac{1}{|\mathcal{C}|}\right)$\\
        DN & $\Theta\left(\frac{\sum_{i}\left(d_i - d^{C}_i\right)^2}{N}\delta^{2}\right)$ & $O\left(\frac{d_{C}^{4}}{|\mathcal{C}|}\right)$ \\
        HT & 0 & $\Theta\left(\frac{2^{d_C}}{|\mathcal{C}|}\right)$
      \end{tabular}
      \caption{Summary of bias and variance results for the DM, DN, and HT estimators in the cluster setting. Here, \(d_{i}\) denotes the degree of node \(i\) in the original graph \(G\), and \(d_{i}^{C}\) represents the number of neighbors of \(i\) that belong to the same cluster as \(i\). \(|\mathcal{C}|\) is the total number of clusters, and \(d_{C} = \max_{i} |\mathcal{N}^{C}_i|\), where \(\mathcal{N}^{C}_i\) is the set of clusters connected to node \(i\). The results illustrate how clustering reduces the bias of the DM estimator by increasing \(d_{i}^{C}\), while simultaneously increasing the variance from \(1/N\) to \(1/|\mathcal{C}|\). Additionally, the DN estimator exhibits similar behavior but with a bias of \(O(\delta^2)\) instead of $O(\delta)$. This highlights how DN introduces a distinct bias-variance tradeoff that can be combined with clustering design.}
      \label{tab:results-summary-cluster}
    \end{table}

\textbf{Proof Sketch.}
The proof follows a Taylor expansion approach similar to Theorem~\ref{dn-bias}, but with a much finer control over the second-order terms. Let \(f(z^{C})\) represent the outcome for a node given a vector of cluster-level treatment assignments \(z^{C} \in \{0, 1\}^{|\mathcal{C}|}\). Define \(F^{1}(w) = {\sf E}[f(z^{C}) | z_{C_i}=1]\) and \(F^{0}(w) = {\sf E}[f(z^{C}) | z_{C_i}=0]\), where the expectation is taken over \(z_{C_{j}} \sim {\rm Bern}(w_j)\) for \(j \neq i\). The key step involves expanding \(F^{1}(p)\) to approximate \(F^{1}(\vec{1})\) and \(F^{0}(p)\) to approximate \(F^{0}(\vec{0})\). The challenge lies in connecting the finite-difference assumption at the unit level (Equation~\ref{assum:second-order-finite-difference}) with the cluster-level bound. Details are deferred to the Appendix. 

Finally, we can bound variance of the DN estimator under cluster-level randomization:

\begin{theorem}
    \label{thm:dn-cluster-variance}
    The variance of the cluster-level DN estimator satisfies:
    \[
    {\rm Var}({\rm DN-Cluster})\leq O\left(\frac{Y_{\max}^2}{|\mathcal{C}|}\cdot\left(d_{C}^4 +\frac{d_{C}^3}{p(1-p)}\right)\right),
    \]
    where \(d_C\) denotes the cluster-level degree of the nodes and \(Y_{\max}\) is the maximum potential outcome.
\end{theorem}
The results are summarized in \cref{tab:results-summary-cluster}.
Equipped with the bias-variance bound, the effectiveness of the cluster-level DN estimator is validated both theoretically through the analysis of small-world graphs and empirically across various practical settings.

\section{Small-World Network: Bias-Variance Analysis}\label{sec:small-world}
In this section, we examine a canonical model in network science: the small-world network \cite{watts1998collective}. We compare various estimators under this model, which is characterized by high clustering coefficients and low average path lengths. Small-world networks have been widely used to model social networks, brain neuron networks, airport networks, and word co-occurrence networks, among others (see \cite{newman2000models} for a review).

\begin{figure}[h]
    \centering
    \begin{minipage}[t]{0.48\textwidth}
    \includegraphics[width=\linewidth]{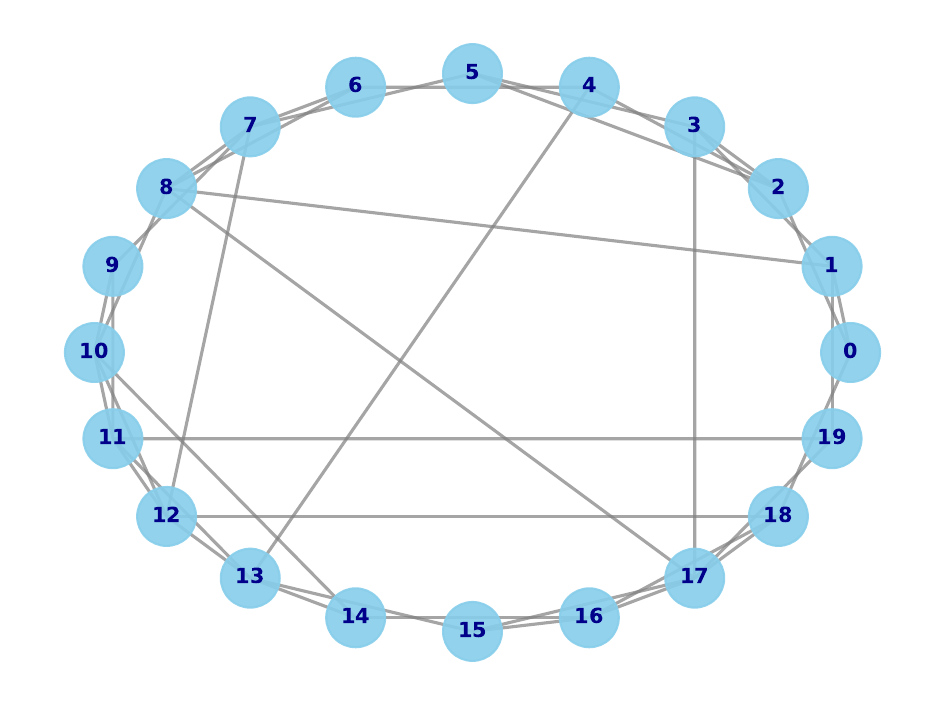}
    \label{fig:small-world-network}
    \end{minipage}
    \begin{minipage}[t]{0.48\textwidth}
    \centering
    \includegraphics[width=\linewidth]{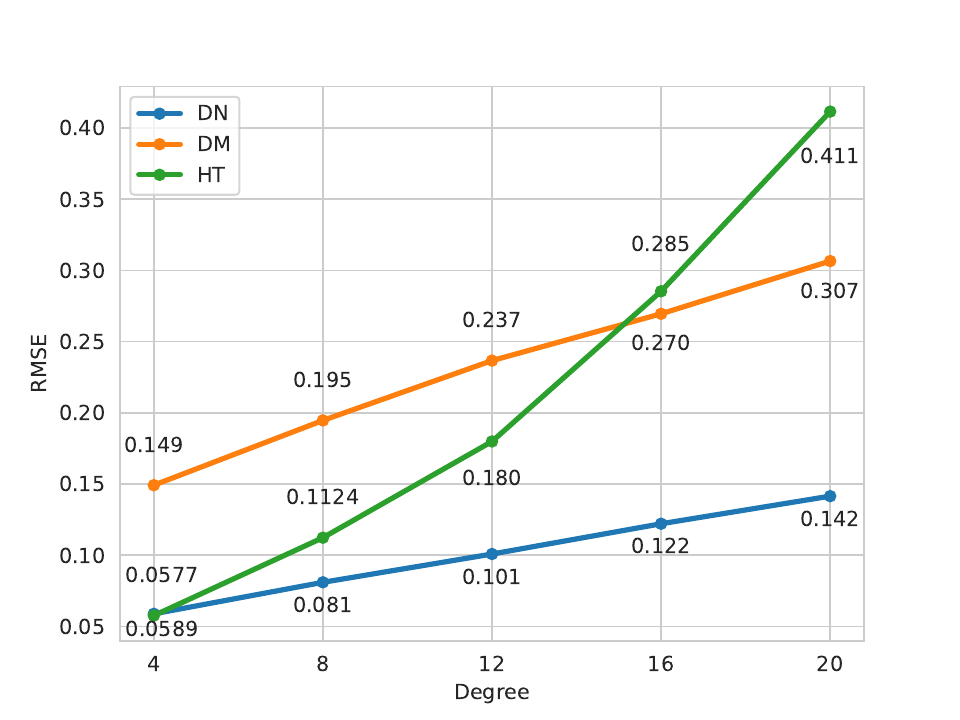}
    \label{fig:optimal-cluster-bias-variance-curve}
    \end{minipage}
    \caption{\textit{Left:} An example of a small-world network with parameters $N=20$, $d=4$, and $q=0.2$. \textit{Right:} RMSE of different estimators while varying degree $d$. The graph has $N = 10000$ with rewiring probability $q = 0.05$. Each point in the plot represents the RMSE corresponding to the optimal cluster size that minimizes the RMSE for each estimator.}
    \label{fig:sw-curve}
\end{figure}

\subsection{Small-World Network Model}
A small-world network $G(N, d, q)$ is generated through the following steps:
\begin{enumerate}
    \item \textbf{Initialization}: Start with a $d$-regular ring graph, where each node is connected to its $d/2$ nearest neighbors on both the left and right.
    \item \textbf{Rewiring}: For each node and each of its rightward $d/2$ edges, with probability $q$, replace the edge with a random connection to another node chosen uniformly at random.
\end{enumerate}

An example of such a network is shown in Fig.~\ref{fig:sw-curve}. When $q=0$, the graph remains a ring with a diameter of $L=O(N)$. When $q=1$, the graph becomes an Erdős–Rényi random graph with a diameter of $L=O(\log(N))$. For small values of $q$, the diameter decreases sharply while the clustering properties remain similar to those of a ring graph, capturing the "small-world" effect observed in real-world networks \cite{watts1998collective}. In our experiments, we consider the canonical scenario where $d=O(\log^{c}(N))$ for $c \geq 1$ and set $p=1/2$ to simplify the analysis.

\subsection{Unit-Level Randomization Design}
We first analyze the performance of the estimators under unit-level randomization. We compare DM, DN, and HT estimators in this setting.
We have
\begin{enumerate}
    \item \textbf{DM Estimator}:
    \begin{itemize}
        \item \textbf{Bias}: The bias of the DM method is $\tilde{O}(\delta)$, as derived in \cref{sec:DN-unit}.
        \item \textbf{Variance}: The variance is bounded by $\frac{Y_{\max}^2}{Np(1-p)}$, which is on the order of $O\left(\frac{Y_{\max}^2}{N}\right)$.
    \end{itemize}

    \item \textbf{DN Estimator}:
    \begin{itemize}
        \item \textbf{Bias}: By Theorem~\ref{dn-bias}, the bias is $\tilde{O}(\delta^2)$.
        \item \textbf{Variance}: By Theorem~\ref{dn-variance}, the variance is bounded by $O\left(\frac{Y_{\max}^2}{N}\cdot \left(\frac{d^3}{p(1-p)}+ d^4\right)\right) = \tilde{O}\left(\frac{Y_{\max}^2}{N}\right)$.
    \end{itemize}

    \item \textbf{HT Estimator}:
    \begin{itemize}
        \item \textbf{Bias}: The bias is $0$ due to propensity score adjustment.
        \item \textbf{Variance}: The variance is $O\left(Y_{\max}^2 N^{\log^{c-1}(N)}\right)$.
    \end{itemize}
\end{enumerate}

The results are summarized in Table~\ref{tab:unit_level_comparison}. From this analysis, the DN estimator demonstrates a favorable bias-variance tradeoff: \textit{compared to DM, it reduces the bias from $O(\delta)$ to $O(\delta^2)$ while only slightly increasing the variance by logarithmic factors. Moreover, its variance is exponentially smaller than that of any unbiased estimator. These results hold for any graph with degree $d$ scaling logarithmically in $N$.}
\begin{table}[h!]
\centering
\caption{Small-World Network: Comparison of Bias and Variance for Unit-Level Randomization. The DN estimator achieves a second-order bias reduction compared to DM. Its variance is exponentially smaller than that of unbiased estimators and only logarithmically larger than that of the DM estimator.}
\label{tab:unit_level_comparison}
\begin{tabular}{|l|c|c|}
\hline
\textbf{Estimator} & \textbf{Bias} & \textbf{Variance} \\
\hline
DM & $\tilde{O}(\delta)$ & $O\left(\frac{Y_{\max}^2}{N}\right)$ \\
\hline
DN & $\tilde{O}(\delta^2)$ & $\tilde{O}\left(\frac{Y_{\max}^2}{N}\right)$ \\
\hline
HT & $0$ & $O\left(Y_{\max}^2 N^{\log^{c-1}(N)}\right)$ \\
\hline
\end{tabular}
\end{table}

\subsection{Clustering Design}

Next, we investigate the impact of clustering on the bias-variance tradeoff for small-world networks. For simplicity, we consider clusters formed by grouping every $m$ adjacent nodes in the ring graph. The results are summarized in Table~\ref{tab:clustering_comparison}.

\begin{enumerate}
    \item \textbf{DM Estimator}:
    \begin{itemize}
        \item \textbf{Bias}: The bias depends on the number of edges crossing clusters. For the ring graph, $\tilde{O}((1-q)\frac{1}{m})$ edges cross clusters, and the new random edges contribute $\tilde{O}(q)$. Thus, the bias is $\tilde{O}(\delta(\frac{1}{m} + q))$ for $q \leq \frac{1}{2}$.
        \item \textbf{Variance}: The variance scales linearly with the inverse of the number of clusters, yielding $O\left(\frac{Y_{\max}^2m}{N}\right)$.
    \end{itemize}

    \item \textbf{DN Estimator}:
    \begin{itemize}
        \item \textbf{Bias}: By Theorem~\ref{thm:dn-cluster-bias}, the bias depends on the out-cluster degree of each node. It is given by:
        \[
        \frac{\sum_{i=1}^{N} (d_i - d_i^{C})^2}{N} = \tilde{O}\left(\frac{1}{m}\delta^2\right) + \tilde{O}\left(q\delta^2\right) = \tilde{O}\left(\left(\frac{1}{m}+q\right)\delta^2\right).
        \]
        \item \textbf{Variance}: By Theorem~\ref{thm:dn-cluster-variance}, the variance is $\tilde{O}\left(\frac{Y_{\max}^2m}{N}\right)$.
    \end{itemize}

    \item \textbf{HT Estimator}:
    \begin{itemize}
        \item \textbf{Bias}: The bias remains $0$ due to propensity score correction.
        \item \textbf{Variance}: The variance grows exponentially with the number of cluster neighbors $d_{C}$, resulting in $\tilde{O}\left(\frac{Y_{\max}^2m N^{q\log^{c-1}(N)}}{N}\right)$. While this bound may theoretically perform well for sufficiently small $q$, it is sensitive to scenarios where a small fraction of nodes have slightly higher degrees. In practice, the HT estimator often exhibits prohibitively large variance, as observed in our experiments.
    \end{itemize}
\end{enumerate}

\textbf{Bias-Variance Tradeoff.}
Compared to the DM method, which is commonly used in practice, the DN estimator with clustering design offers a superior bias-variance tradeoff. For example, when $q=0$, the optimal cluster size $m$ to minimize the RMSE of the DM estimator is $m^{*} := (\delta^2 N)^{1/3}$, yielding:
\[
RMSE^{*}(DM) = \frac{\delta^{1/3}}{N^{1/3}}.
\]

In contrast, the optimal cluster size for the DN estimator is $m^{*} = (\delta^{4}N)^{1/3}$, resulting in:

\[
RMSE^{*}(DN) = \frac{\delta^{2/3}}{N^{1/3}} \ll RMSE^{*}(DM).
\]

By selecting a smaller cluster size, the DN estimator achieves an RMSE that breaks the fundamental limit of clustering with DM estimators. Experimentally we optimize the cluster size for various estimators and the RMSE curve is shown in Figure~\ref{fig:sw-curve} which clearly demonstrates the bias-variance frontier established by DN. 


\begin{table}[h!]
\centering
\caption{Small-world Network: Comparison of Bias and Variance for Clustering Design. Every cluster consists of $m$ adjacent nodes. DN is exhibiting a new frontier of bias-variance curve when optimizing the cluster design, see experiments in Figure~\ref{fig:sw-curve}.}
\label{tab:clustering_comparison}
\begin{tabular}{|l|c|c|}
\hline
\textbf{Estimator} & \textbf{Bias} & \textbf{Variance} \\
\hline
DM (Cluster size $m$) & $\tilde{O}\left(\frac{1}{m}\delta + q\delta \right)$ & $O\left(\frac{Y_{\max}^2 m}{N}\right)$ \\
\hline
DN (Cluster size $m$) & $\tilde{O}\left(\frac{1}{m}\delta^2 + q\delta^2 \right)$ & $\tilde{O}\left(\frac{Y_{\max}^2 m}{N}\right)$\\ 
\hline
HT (Cluster size $m$) & $0$ & $O\left(\frac{Y_{\max}^2 m N^{q\log^{c-1}(N)}}{N}\right)$\\
\hline
\end{tabular}
\end{table}

\section{Experiments}
\label{sec:experiments}
A key contribution of this work is the superior practical performance of the DN estimator on experiments with realistic structures and scale. Here, we demonstrate this via a series of experiments comparing DN to the incumbent DM and HT estimators, across multiple graph structures and outcome models. 

We first study synthetic outcome models under various random graph structures (Watts Strogatz small world graph; Erdos-Renyi graph) as well as a real social network on Twitter \cite{twitter_network}). These experiments demonstrate the performance of DN both under unit-level Bernoulli randomization and under cluster randomization. Next, we evaluate DN in a realistic city-scale simulator of a ridesharing platform, with no closed-form ground truth outcome model, demonstrating similar conclusions. Taken together, these results demonstrate the following:

\begin{enumerate}
    \item The DN estimator achieves large bias reductions relative to DM and large variance reductions relative to HT, both with and without clustering; the end result is substantially reduced error compared to either incumbent.
    \item DN estimation works together with clustering to reduce bias. As a result, optimal cluster sizes are smaller (thus more clusters) when DN is used for estimation, as compared with DM. 
\end{enumerate}

\textbf{Clustering:} Throughout this section, we construct clusters using the CPM clustering scheme \cite{CPM}, and study various settings of the ``resolution'' parameter, which essentially controls the number of clusters. We refer the reader to \cite{CPM} for details.

\subsection{DN outperforms clustering with DM}
We first study the benefits of DN in a setting where direct interference from neighbors is large, but indirect interference due to interactions between neighbors is comparatively small. In particular, we consider the following outcome model: 
\begin{equation}
    \label{eqn:experiment}
    f_i(z) = c_{0}\sum_{j=\mN_i \cup\{i\}} \frac{z_iz_j}{|\mN_i|} + c_1\prod_{j\in \mN_i\cup\{i\}}\left(1+c_2 z_j\right)
\end{equation}
Note that based on this outcome model, the finite second order difference (Equation~\ref{assum:second-order-finite-difference}) is controlled exactly by $c_2$, since $\Delta f_{jk} = (1+c_2)^2 - 2(1+c_2) +1 = c_2^2$. This parameter thus scales the smoothness of the outcome function. Under both the Erdos-Renyi random graph and a small world graph setting, we perform estimation for number of nodes $N=5000,10000,15000$, representing the individuals in a network. For each $N$, we run $5$ random graphs each for $1000$ trials. We compare DN under unit-level Bernoulli treatment with DM across different clusterings generated by \cite{CPM} at varying resolution levels. 

Although the Horvitz-Thompson estimator is unbiased in this setting, its variance is excessively high under both unit-level Bernoulli designs and cluster-based Bernoulli designs. Due to this instability, we exclude it from our experimental results. We measure the performance of various estimators by the relative error and the corresponding RMSE.
\begin{align*}
    \text{relative error} = \frac{1}{K}\sum_{k=1}^K\frac{\hat{\rm ATE}_k - {\rm ATE}}{{\rm ATE}}, \quad \text{RMSE} =  \sqrt{\frac{1}{K}\sum_{k=1}^K\left(\hat{\rm ATE}_k - {\rm ATE}\right)^2}
\end{align*}
where ${\rm \hat{ATE}}_k$ is the estimator of trial $k$ and ${\rm ATE}$ is the estimand.

Figure~\ref{fig:sw_unbiased} summarizes the results of the Small World graph, and Figure~\ref{fig:er-unbiased} summarizes the results of the Erdos-renyi results.

\begin{figure}[h]
    \centering
    \begin{minipage}[b]{0.33\textwidth}
        \includegraphics[scale = 0.63]
        {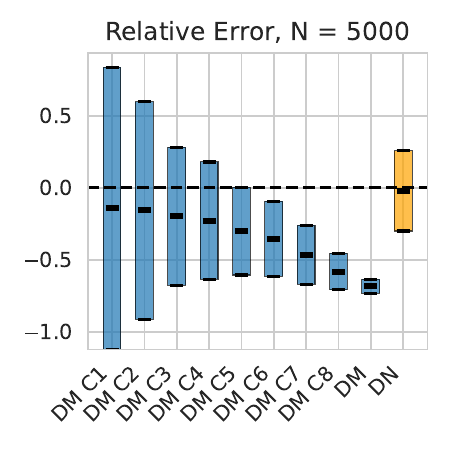}
    \end{minipage}\hfill
    \begin{minipage}[b]{0.33\textwidth}
        \includegraphics[scale = 0.63]{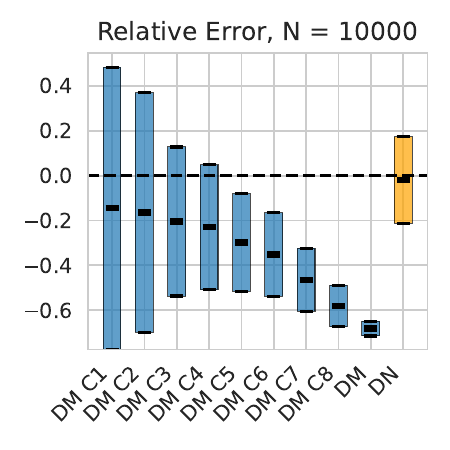}
    \end{minipage}\hfill
    \begin{minipage}[b]{0.33\textwidth}
        \includegraphics[scale = 0.63]{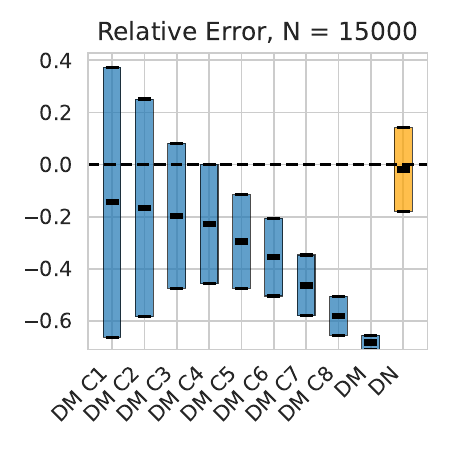}
    \end{minipage}
    
    \caption{Relative error with 95\% confidence interval for small world random graph. We take the degree of the starting ring graph $d=20$, and set the rewiring probability $q= 0.1$. Again clusters $1$ to $8$ are produced with CPM clustering with resolution $0.005, 0.01, 0.05, 0.1,0.3, 0.5, 0.7, 0.9$. The rough size of the clusters ranges from $<10^2$ for resolutions $<0.1$, to around $10^2$ going up to $\sim600$ for resolutions $<0.7$, $10^3$ for resolution $0.7,0.9$, and Node size for $1.0$ and above.}
    \label{fig:sw_unbiased}
\end{figure}

\begin{figure}[h]
    \centering
    \begin{minipage}[b]{0.33\textwidth}
        \includegraphics[scale = 0.63]
        {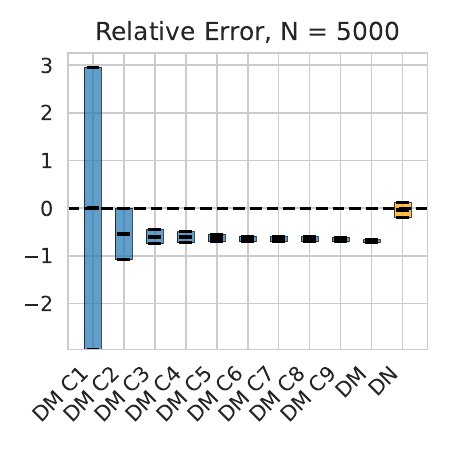}
    \end{minipage}\hfill
    \begin{minipage}[b]{0.33\textwidth}
        \includegraphics[scale = 0.63]{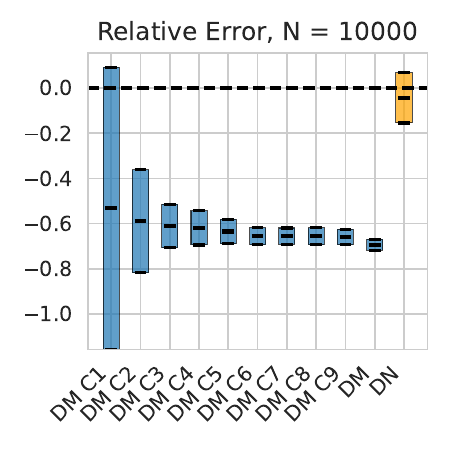}
    \end{minipage}\hfill
    \begin{minipage}[b]{0.33\textwidth}
        \includegraphics[scale = 0.63]{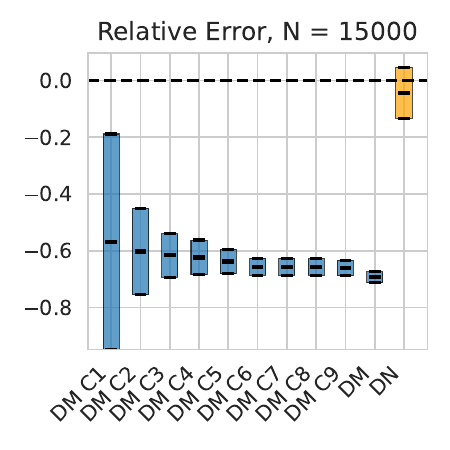}
    \end{minipage}
    
    \caption{Relative bias with 95\% confidence interval for Erdos-Renyi random graph. We take $p_{deg} = 20/N$, meaning the expected degree is $20$. Clusters $1$ to $9$ are produced with CPM clustering with resolution $0.005 0.01, 0.05, 0.1,0.3, 0.5, 0.7, 0.9, 1.0$.}
    \label{fig:er-unbiased}
\end{figure}

\textbf{Twitter network:} We also conducted experiments on a real-world social network with $N = 81,306$ nodes and $E = 1,768,149$ edges \cite{twitter_network}. To better simulate real-world conditions, we add on an additional $\epsilon\sim N(0,c)$ noise to the outcome function. The results are summarized in Figure~\ref{fig:social network}.

\begin{figure}[h]
\begin{minipage}{0.99\textwidth}
    \centering
    \begin{minipage}[b]{0.48\textwidth}
        \centering%
        \includegraphics[scale = 0.47]{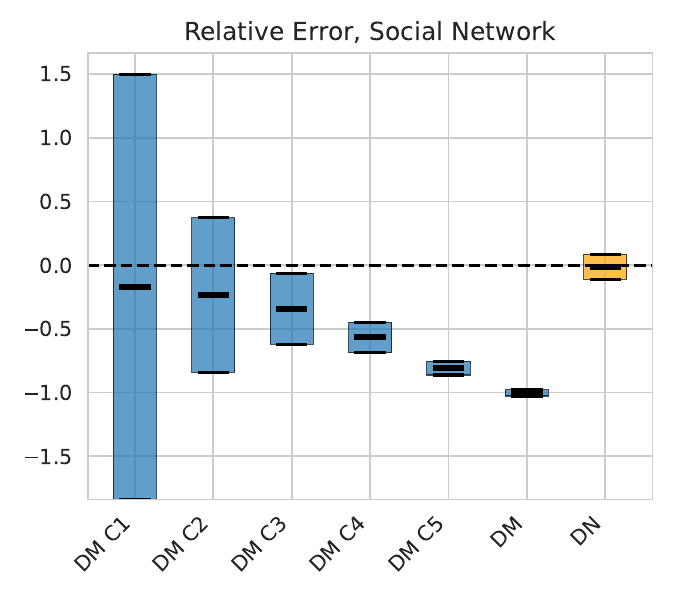}%
    \end{minipage}%
    \hfill
    \begin{minipage}[b]{0.44\textwidth}
        \centering%
            \begin{tabular}{lcc}
                \hline
                \textbf{Estimator} & \textbf{Clusters} &  \textbf{RMSE}  \\
                \hline
                DN (Unit-level) & $81306$ & ${\bf 0.06}$ \\
                DM (Unit-level) & $81306$ & $1.22$  \\
                DM (Cluster 1) & $3967$ & $1.20$\\
                DM (Cluster 2) & $6003$ & $0.59$\\
                DM (Cluster 3) & $13296$ & $0.52$\\ 
                DM (Cluster 4) & $24388$ &  $0.71$\\
                DM (Cluster 5) & $45299$ & $0.99$ \\
                \hline
                \end{tabular}
    ~\newline
    \vspace{0.5cm}
    \end{minipage}%
\end{minipage}
    \caption{Real twitter network graph with synthetic outcome function. Clusters 1 through 5 correspond to resolution  $0.0001$, $0.001$, $0.01$, $0.1$, $0.5$. \textit{Left:} Relative error across $100$ trials. \textit{Right:} RMSE and clustering information. DN's performance dominates DM regardless of the clustering scheme. }
    \label{fig:social network}
\end{figure}

\textbf{DN achieves a superior bias-variance tradeoff:} Our estimator significantly reduces bias compared to DM, with only a modest increase in variance. While clustering helps DM mitigate bias, it does so at the expense of a substantially higher variance. These findings are consistent across different graph sizes and remain robust in Erdos-Rényi settings.

\subsection{DN with clustering}
We next examine a less benign setting, in which we increase the bias of DN by increasing $c_2$, thereby increasing the bias upper bound from \cref{th:general-bias}. We show the results for this setting in \cref{fig:sw_bias}, where we demonstrate that:

\begin{enumerate}
    \item  \textbf{DN with unit randomization achieves lower RMSE than DM with any clustering scheme}, as before, albeit with substantially larger bias than in the previous setting.
    \item  \textbf{Combining DN with clustering achieves superior bias-variance tradeoffs:} While DN with unit randomization already outperforms DM, the question remains whether good experimental {\it design} can further improve DN's performance in parameter regimes such as this one, where the bias guarantee provided by \cref{th:general-bias} is relatively weak. We demonstrate here that this is indeed the case: as analyzed in Section~\ref{sec:dn-cluster}, for both DM and DN estimators, reducing the number of clusters from $N$ (where each node forms its own cluster) decreases bias while increasing variance, introducing an additional dimension to the bias-variance tradeoff. As shown in the right plot of Figure~\ref{fig:sw_bias}, this tradeoff produces a characteristic curve with respect to cluster sizes. Notably, at cluster level 3 (resolution $0.1$), DN achieves the lowest RMSE among all DN and DM clustering schemes. 
\end{enumerate}

\begin{figure}[h]
    \centering
    \scalebox{0.88}{
    \begin{minipage}[b]{0.5\textwidth}
        \centering
        \includegraphics[width=\textwidth]{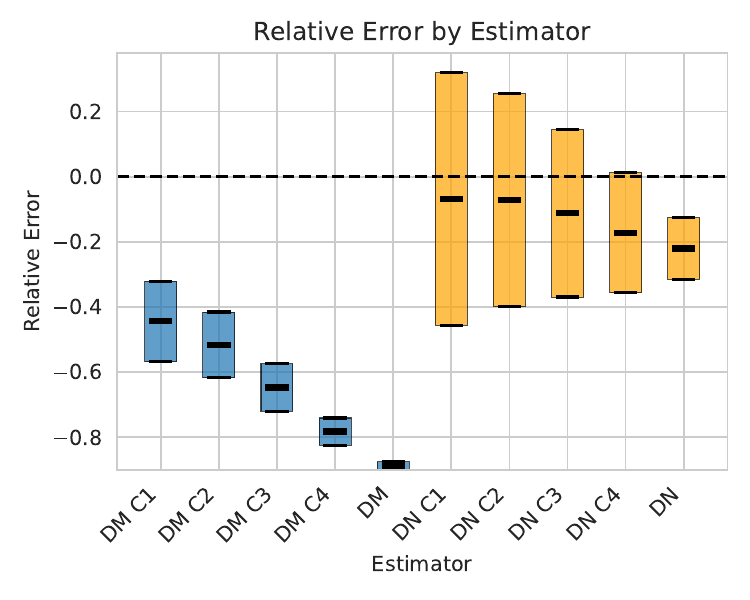}
        \label{fig:er}
    \end{minipage}
    \hspace{0.5cm}
    \begin{minipage}[b]{0.45\textwidth}
        \centering
        \includegraphics[width=\textwidth]{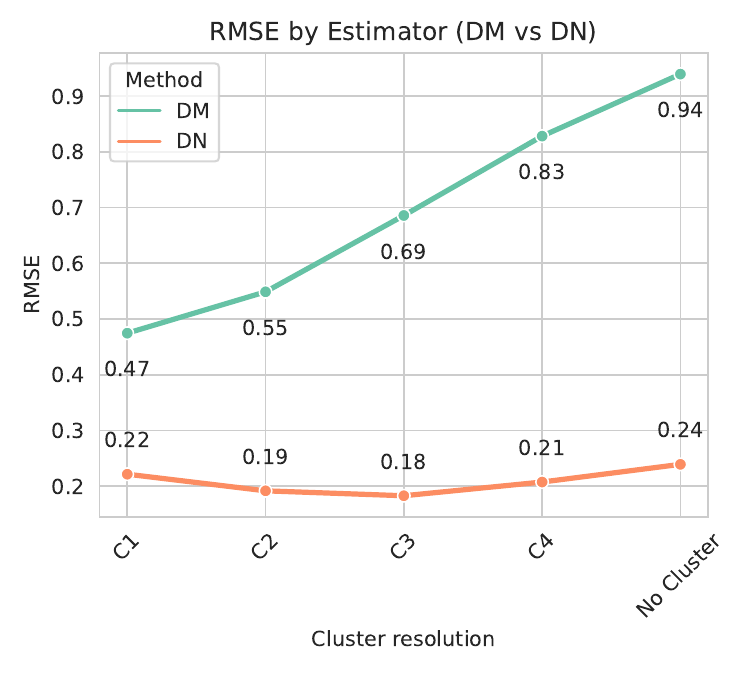}
        \label{fig:sw}
    \end{minipage}
    }
    \vspace{-0.5cm}
    \caption{Small World random graph with $N = 15000$. We set $c_2 = 0.2$ as high order interference and the ground truth ATE is $1.06$. C1, C2, C3, C4 corresponds to resolution $0.3, 0.5,0.7,0.9$, with number of clusters roughly around $1/2\cdot 10^2$, $1/2\cdot 10^2$,  $1.5\cdot10^3$, $2.5\cdot 10^3$. \textit{Left:} Relative error with 95\% confidence interval by estimator. \textit{Right:} RMSE by estimator. DN dominates DM for all clustering schemes, and the optimal bias-variance tradeoff is achieved at cluster 3, with resolution $0.1$. This cluster achieves roughly 25\% increase in RMSE from unit DN. }
    \label{fig:sw_bias}
\end{figure}

\subsection{Ridesharing Simulation}

\begin{figure}[h]
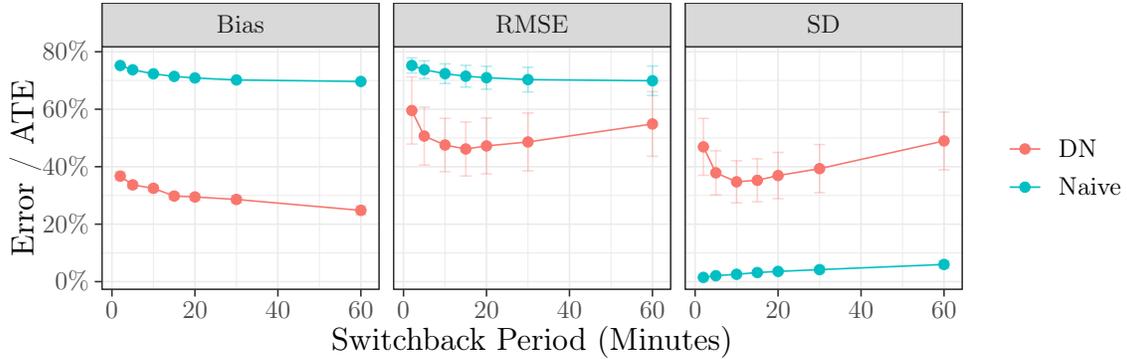

  \centering%
  \include{ridesharing}%
  \vspace{-1cm}%
  \caption{Relative error, standard deviation, and RMSE of DN and DM estimators, applied to estimating the treatment effect of a pricing algorithm change in a ride-sharing simulator, as a function of the temporal cluster size (i.e., switchback duration). DN achieves lower RMSE as compared with DM for every cluster size, with an optimal cluster size of around 15 minutes.}%
  \label{fig:ridesharing}
\end{figure}

In this section, we demonstrate the performance of the DN estimator with clustering in a real-world setting, in which there is no closed-form outcome function and the assumptions made to bound the bias of DN are not necessarily satisfied a priori.

To this end, we simulate a pricing experiment in a real-world ride-sharing simulator.\footnote{We have made the simulation code available at \url{https://github.com/atzheng/or-gymnax}} The system observes a series of ``eyeballs'' (i.e., potential ride requests), in which a rider inputs their pickup and dropoff locations, and the rider is presented with a posted price for their trip, as well as an estimated time to pickup (ETA). Based on this information, the rider can choose to accept the trip, in which case the nearest car is dispatched and the system receives the price of the trip as a reward; or they can reject it, in which case no dispatch is made and the reward is zero. The rider's acceptance decision is based on a simple logistic choice model:

\begin{equation*}
P({\rm Accept}_{i}) = \frac{1}{1 + \exp(- (\beta_{0} + \beta_{\rm price} \cdot {\rm Price} + \beta_{\rm ETA} \cdot {\rm ETA}))}
\end{equation*}
where $\beta_{\rm ETA}, \beta_{\rm price} < 0$.

For simplicity, the price offered to the rider is simply a price per unit time, multiplied by the time required to drive the rider from their pickup location to their destination. The intervention in the experiment is to increase this price, which has the direct effect of changing the expected revenue for the trip, as well as a downstream interference effect of modifying supply availability. In particular, the naive DM estimator systematically underestimates the benefits of a price increase: while increasing the price myopically reduces the probability of acceptance, this increases future supply availability and therfore acceptance probabilities at future times.

As is common in online platforms, here we employ a spatio-temporally clustered experiment design: all eyeballs occuring within the same time period and geographic region recieve prices according to the same pricing policy, either treatment or control. Spatial clusters are taken to be the taxi zones defined by New York City's Taxi and Limousine Commission (TLC), which partitions Manhattan into roughly 60 zones. We vary the size of temporal clusters from two minutes to one hour. Anecdotally, switchback experiments in service platforms use switchback durations at the upper end of this range \cite{SwitchbackTestsRandomized, bojinovDesignAnalysisSwitchback2023}.

Eyeballs (request time, pickup location, and dropoff location) are taken from TLC data\footnote{https://www.nyc.gov/site/tlc/about/tlc-trip-record-data.page}, restricted to Manhattan. Cars are routed along the shortest path according to the street grid provided by OpenStreetMaps. At the beginning of the period, car locations are randomly initialized to locations on the street grid. We run our experiment on 500,000 eyeballs, or about one week of data.

In order to implement DN estimation, we need to assume a particular interference network. In reality, all eyeballs have some non-zero causal impact on all other eyeballs; however, we will assume for the purposes of estimation that the interference effect of eyeballs separated by time and distance above a certain threshold is neglible. In practice, this threshold constitutes a hyperparameter which provides another lever for bias variance tradeoffs. For this experiment, we find that a time threshold of 10 minutes and a distance threshold of 2 km works well.

We present the results from this experiment in Figure~\ref{fig:ridesharing}. The conclusions here are largely similar to our synthetic instances. The bias of the Naive DM estimator decreases monotonically as we increase the size of the temporal clusters (i.e., the switchback duration). However, the bias stops decreasing at a certain threshold, remaining above 70\% for all cluster sizes, which we attribute to unaccounted-for spatial interference. The error of HT is sufficiently large that we do not include it in the figure, for scale. The DN estimator outperforms both alternatives at all cluster sizes, primarily by sharply reducing bias despite increased variance. Here we find that DN has a relatively small optimal cluster size of about 15 minutes as opposed to one hour for DM. With this clustering scheme, DN reduces bias relative to DM by more than 50\%, while reducing RMSE by about 35\%.

%

~\newline
\noindent\textbf{\large{Acknowledgments}} \hspace{0.5em} We thank Ramesh Johari for his valuable comments.

\newpage
\bibliographystyle{abbrvnat}

\setlength{\bibsep}{.7em}
\bibliography{reference}

\appendix
\section{Proofs}
\subsection{Proof of Theorem~\ref{dn-variance}}
\label{dn-variance-proof}
\begin{proof}
Denote the propensity as $\eta$ and the propensity score correction as $\xi$, i.e.
\[\eta_{i} = \frac{1(z_i=1)}{p} - \frac{1(z_i=0)}{1-p}\qquad \xi_{i} = \frac{1(z_i=1)(1-p)}{p} + \frac{1(z_i=0)p}{1-p}\]
Simple calculation gives us:
\begin{equation}
\label{basic-property-variance}
    \mathbb{E}[\eta] = 0
    \quad \mathbb{E}[|\eta|] = 2 \quad \mathbb{E}[\xi] = 1\quad \mathbb{E}(\eta^2) = \frac{1}{p(1-p)}\quad \mathbb{E}(\xi^2) = \frac{3p^2-3p+1}{p(1-p)} 
\end{equation}
where the expectation is taken over assignment $z$. For notational simplicity, we let $q = \frac{1}{p(1-p)}$. Expanding out the variance, we have
    \begin{align*}
        {\rm Var}({\rm DN}) &= {\rm Var}\left(\frac{1}{N}\sum_{i}^NY_i \left(\eta_i + \xi_i \sum_{j\in \mN_i}\eta_j\right)\right)\\
        & = \frac{1}{N^2}\sum_{i}^N\sum_{j}^N{\rm cov}\left(Y_i\left(\eta_i+\xi_i\sum_{k\in N_i}\eta_k\right),Y_j\left(\eta_j+\xi_j\sum_{l\in N_j}\eta_l\right)\right)
    \end{align*}
    We partition the set of $i,j$ into the following cases:
    \begin{enumerate}
        \item $i = j$
        \begin{align*}
            &{\rm cov}\left(Y_i\eta_i+Y_i\xi_i\sum_{k\in \mN_i}\eta_k, Y_i\eta_i+Y_i\xi_i \sum_{l\in \mN_i}\eta_l\right)\\
            &= {\rm cov}\left(Y_i\eta_i,Y_i\eta_i\right) + 2{\rm cov}\left(Y_i\eta_i,Y_i\xi_i\sum_{k\in \mN_i}\eta_k\right)+ {\rm cov}\left(Y_i\xi_i\sum_{k\in \mN_i}\eta_k, Y_i\xi_i \sum_{l\in\mN_i}\eta_l\right)\\
            &\leq Y_{\max}^2 \mathbb{E}[\eta_i^2] + 2\sum_{k\in \mN_i}\left(\mathbb{E}[Y_i^2\xi_i\eta_i\eta_k] - \mathbb{E}[Y_i\eta_i]\mathbb{E}[Y_i\xi_i\eta_k]\right) \\
            &\hspace{1cm}+ \sum_{k\in \mN_i, l\in \mN_i, k\neq l}\left(\mathbb{E}[Y_i^2\xi_i^2\eta_k\eta_l]- \mathbb{E}[Y_i\xi_i\eta_k]\mathbb{E}[Y_i\xi_i\eta_l]\right) + \sum_{k\in \mN_i}\left(\mathbb{E}[Y_i^2\xi_i^2\eta_k^2] - \mathbb{E}[Y_i\xi_i\eta_k]^2\right)\\
            &= \frac{Y_{\max}^2}{p(1-p)} + 2Y_{\max}^2\sum_{k\in \mN_i}\left(\mathbb{E}[\xi_i|\eta_i||\eta_k|] + \mathbb{E}[|\eta_i|]\mathbb{E}[\xi_i|\eta_k|]\right) \\
            &\hspace{1cm} + Y_{\max}^2 \sum_{k,j\in \mN_i, k\neq j}\left(\mathbb{E}[\xi_i^2|\eta_k||\eta_l|] + \mathbb{E}[\xi_i |\eta_k|]\mathbb{E}[\xi_i|\eta_l|]\right) +Y_{\max}^2\sum_{k\in \mN_i}\left(\mathbb{E}[\xi_i^2\eta_k^2] \right)\\
            &\leq \frac{Y_{\max}^2}{p(1-p)} + 2Y_{\max}^2\sum_{k\in \mN_i}\left(\frac{2p^2-2p+1}{p(1-p)}\cdot2 + 2\cdot 1\cdot 2\right) \\
            &\hspace{1cm} + Y_{\max}^2 \sum_{k,j\in \mN_i, k\neq j}\left(\frac{3p^2-3p+1}{p(1-p)} \cdot2\cdot2+ 2\cdot2\right) + Y_{\max}^2\sum_{k\in \mN_i}\left(\frac{3p^2-3p+1}{p(1-p)}\cdot\frac{1}{p(1-p)} \right)
        \end{align*}
        where the first equality is due to the bilinarity of covariance, the last inequality is by plugging in Equation~\ref{basic-property-variance}, and the rest is by plugging in covariance definition and taking out the $Y$ terms.
    
        Therefore we have in the case of $i=j$,
        \begin{align*}
            &{\rm cov}({\rm DN}_{ij}) \leq Y_{\max}^2 \Bigl(\frac{1}{p(1-p)} + d_i\cdot \frac{4}{p(1-p)} \\
             &\hspace{3cm}+(d_i^2-d_i)\cdot(\frac{4}{p(1-p)}-8) + d_i\frac{1}{p(1-p)}(\frac{1}{p(1-p)}-3)\Bigr)\\
            &= Y_{\max}^2 \left(q+ 5d_iq +4d_i^2q -8d_i^2+8d_i +d_iq^2\right)
        \end{align*}
        
        \item $i\neq j$ and $i$ and $j$ are neighbors.
        \begin{align*}
            &{\rm cov}\left(Y_i\eta_i+Y_i\xi_i\sum_{k\in \mN_i}\eta_k, Y_j\eta_j+Y_j\xi_j \sum_{l\in \mN_j}\eta_l\right)\\
            &= {\rm cov}\left(Y_i\eta_i,Y_j\eta_j\right) + {\rm cov}\left(Y_i\eta_i,Y_j\xi_j \sum_{l\in \mN_j}\eta_l\right) \\&\hspace{2.5cm}+ {\rm cov}\left(Y_i\xi_i\sum_{k\in \mN_i}\eta_k, Y_j\eta_j\right) + {\rm cov}\left(Y_i\xi_i\sum_{k\in \mN_i}\eta_k, Y_j\xi_j \sum_{l\in \mN_j}\eta_l\right)
        \end{align*}
        \begin{enumerate}
            \item First term:
            \begin{align*}
                &{\rm cov}\left(Y_i\eta_i,Y_j\eta_j\right) \leq Y_{\max}^2 \mathbb{E}[|\eta_i||\eta_j|] + Y_{\max}^2 \mathbb{E}[|\eta_i|]\mathbb{E}[|\eta_j|] = Y_{\max}^2 \cdot 8
            \end{align*}
            \item Second term:
            \begin{align*}
                {\rm cov}\left(Y_i\eta_i,Y_j\xi_j \sum_{l\in \mN_j}\eta_l\right) &= \sum_{l\in \mN_j, l\neq i}{\rm cov}\left(Y_i\eta_i,Y_j\xi_j\eta_l\right) + {\rm cov}\left(Y_i\eta_i,Y_j\xi_j\eta_i\right)\\
                &\leq 8Y_{\max}^2 (d_j-1) + Y_{\max}^2\cdot (\frac{1}{p(1-p)} + 4)\\
                &= Y_{\max}^2\left(8d_j -4+ q\right)
            \end{align*}
            \item Third term:
            \begin{align*}
                {\rm cov}\left(Y_i\xi_i\sum_{k\in \mN_i}\eta_k, Y_j\eta_j\right) \leq Y_{\max}^2\left(8d_i -4+ q\right)
            \end{align*}
            \item Fourth term:
            \begin{align*}
                {\rm cov}&\left(Y_i\xi_i\sum_{k\in \mN_i}\eta_k, Y_j\xi_j \sum_{l\in \mN_j}\eta_l\right) \\
                &= \sum_{k\in \mN_i, l\in \mN_j, k\neq j, l\neq i}{\rm cov}\left(Y_i\xi_i\eta_k, Y_j\xi_j\eta_l\right) + \sum_{k = j, l\neq i}{\rm cov}\left(Y_i\xi_i\eta_k, Y_j\xi_j\eta_l\right)\\
                &\hspace{1cm}+ \sum_{k\neq j, l = i}{\rm cov}\left(Y_i\xi_i\eta_k, Y_j\xi_j\eta_l\right)+ {\rm cov}\left(Y_i\xi_i\eta_j, Y_j\xi_j\eta_i\right)\\
                &\leq \sum_{k\neq j, l\neq i} Y_{\max}^2\left(\E[|\xi_i\xi_j\eta_k\eta_l|] +\E[|\xi_i\eta_k|]\E[|\xi_j\eta_l|]\right)\\
                &\hspace{1cm}+ \sum_{k = j, l\neq i} Y_{\max}^2\left(\E[|\xi_i\xi_j\eta_j\eta_l|] +\E[|\xi_i\eta_j|]\E[|\xi_j\eta_l|]\right)\\
                &\hspace{1cm}+ \sum_{k\neq j, l = i} Y_{\max}^2\left(\E[|\xi_i\xi_j\eta_k\eta_i|] +\E[|\xi_i\eta_k|]\E[|\xi_j\eta_i|]\right)\\
                &\hspace{1cm}+ \sum_{k = j, l = i} Y_{\max}^2\left(\E[|\xi_i\xi_j\eta_j\eta_i|] +\E[|\xi_i\eta_j|]\E[|\xi_j\eta_i|]\right)\\
                &= Y_{\max}^2\left(8(d_i-1)(d_j-1) + (d_j-1)(2\cdot \frac{2p^2-2p+1}{p(1-p)}+ 4)\right) \\
                &\hspace{1cm} + Y_{\max}^2\left((d_i-1)(2\cdot \frac{2p^2-2p+1}{p(1-p)}+ 4) + (\frac{2p^2-2p+1}{p(1-p)})^2 +4\right)
            \end{align*}
        \end{enumerate}
        Collecting all the terms, we have for $i\neq j$ and $i$ and $j$ are neighbors,
        \begin{align*}
            &{\rm cov}({\rm DN}_{ij}) \leq Y_{\max}^2\left(8d_id_j + 2qd_i+2qd_j +q^2-6q+16\right)
        \end{align*}
        \item $i\neq j$ and they are not neighbors, but $i$ and $j$ share common neighbors.
        \begin{align*}
            &{\rm cov}\left(Y_i\eta_i+Y_i\xi_i\sum_{k\in \mN_i}\eta_k, Y_j\eta_j+Y_j\xi_j \sum_{l\in \mN_j}\eta_l\right)\\
            &= {\rm cov}\left(Y_i\eta_i,Y_j\eta_j\right) + {\rm cov}\left(Y_i\eta_i,Y_j\xi_j \sum_{l\in \mN_j}\eta_l\right) \\
            &\hspace{2.5cm}+ {\rm cov}\left(Y_i\xi_i\sum_{k\in \mN_i}\eta_k, Y_j\eta_j\right) + {\rm cov}\left(Y_i\xi_i\sum_{k\in \mN_i}\eta_k, Y_j\xi_j \sum_{l\in \mN_j}\eta_l\right)\\
            &\leq Y_{\max}^2\left(8d_i+8d_j +8\right) + \sum_{k\neq l}{\rm cov}\left(Y_i\xi_i\eta_k, Y_j\xi_j \eta_l\right) + \sum_{k =  l}{\rm cov}\left(Y_i\xi_i\eta_k, Y_j\xi_j \eta_k\right)\\
            &\leq Y_{\max}^2\left(8d_i+8d_j +8 \right) + Y_{\max}^2\sum_{k\neq l} 8 + Y_{\max}^2 \sum_{k = l}( q +4)\\
            &\leq Y_{\max}^2\left(8d_i+8d_j +8 \right) + Y_{\max}^2 8 d_id_j  + Y_{\max}^2\text{Min}(d_i,d_j)( q +4)
        \end{align*}
        \item $i$ and $j$ do not share common neighbors, in this case the covariance is zero.
    \end{enumerate}
    Summing up all the cases, we have
    \begin{align*}
        {\rm Var}({\rm DN}) &\leq \frac{Y_{\max}^2}{N^2} \sum_{i=1}^N \Bigl(q + 5d_iq +4d_i^2q -8d_i^2+8d_i +d_iq^2 \\
        &\hspace{3cm}+\sum_{j\in \mN_i}\left(8d_id_j + 2qd_i+2qd_j +q^2-6q+16\right) \\
        &\hspace{3cm}+\sum_{j\in M_i\cap\mN_i}(8d_id_j + 8d_i+8d_j+8 + (q+4)Min(d_i,d_j))\Bigr)
    \end{align*}
    Therefore, if we let $d$ denote the upper bound to all the node degrees, we have that
    \begin{align*}
        {\rm Var}(DN) \leq \frac{Y_{\max}^2}{N}\left(8d^4 + qd^3+ 20d^3 + 7qd^2 -20d^2 +q^2d +16d+q\right)
    \end{align*}
    Hence on the order of $O\left(\frac{Y_{\max}^2}{N}\cdot \left(d^4+\frac{d^3}{p(1-p)}+ \frac{d}{p^2(1-p)^2}\right)\right)$
\end{proof}

\subsection{Proof of Theorem~\ref{thm:dn-cluster-bias}}
Define $\mN_i^C$ as the cluster neighbors of node $i$, and let $C_i$ denote the cluster that contains $i$ itself. Note that here we adopt the definition of $\mN_i^C$ such that $C_i\notin \mN_i^C$. We first prove the following lemma that relates cluster level second order finite difference to unit level second order finite difference.  
\begin{lemma}
    Let $C_j,C_k$ denote two neighbor clusters of $i$, and $C_j,C_k \neq C_i$. Suppose that the two clusters  contain $m,n$ number of neighbors of $i$ respectively, i.e. $|C_j\cap \mN_i| = m, |C_k \cap \mN_i| = n$, then
    \[f^a_i(z^{1_m, 1_n}) - f_i^a(z^{1_m,0_n}) - f_i^a(z^{0_m,1_n}) + f_i^a(z^{0_m,0_n}) \leq mn\cdot L \delta^2\]
    for $a\in\{0,1\}$ and any $z\in \{0,1\}^{\mN_i}$, where $1_m, 0 _m$ and $1_n,0_n$ denote the $1$ and $0$ vectors of dimensions $m$ and $n$.
\end{lemma}
\begin{proof}
    For notational simplicity, we drop the $z$ base and directly write out the treatment vector of $m$ and $n$. We prove the statement by induction on $m,n$ jointly. In the case of $m = 1, n=1$, the statement is clearly true by definition. We now assume that for any $ m\leq M, n\leq N$, the statement holds true. We first prove the induction step on $N+1$: that the statement holds true for any $m \leq M, n = N+1$. For any $m = 1,\dots, M$ we have 
    \begin{align*}
        &f_i^a(1_m,1_N,1) - f_i^a(1_m, 0_N, 0) - f_i^a(0_m, 1_N,1) + f_i^a(0_m, 0_N, 0)\\
        &\quad  = \left(f_i^a(1_m, 1_N,1) - f_i^a(1_m, 0_N,1) - f_i^a(0_m, 1_N,1) + f_i^a(0_m,0_N, 1)\right) \\
        &\quad \quad + \left(f_i^a(1_m,0_N,1) - f_i^a (1_m,0_N,0) - f_i^a(0_m,0_N,1) + f_i^a(0_m,0_N,0)\right)\\
        &\quad \leq m\cdot N\cdot L\delta^2 + m\cdot 1\cdot L \delta^2\\
        &\quad = m(N+1)\cdot L \delta^2
    \end{align*}
    where the inequality is due to induction step $m\cdot N$ and $m\cdot 1$ respectively. Now we can complete the induction step by assuming the statement holds true for any $m\leq M, n\leq N+1$ and induct on $m = M+1$, $n = 1,\dots, N+1$. This case follows from the exact same proof as above. Hence we have that the statement holds true for any $m\leq M+1,n\leq N+1$. The induction is complete.
\end{proof}
We can then follow the exact same proof as Theorem~\ref{dn-bias} with the cluster level function $f_i^C: z^C\to \mathbb{R}$. To recap, we let \(f_i(z^{C})\) represent the outcome for node $i$ given a vector of cluster-level treatment assignments \(z^{C} \in \{0, 1\}^{|\mathcal{C}|}\). Define \(F_i^{1}(w) = {\sf E}[f_i^C(z^{C}) | z_{C_i}=1]\) and \(F_i^{0}(w) = {\sf E}[f_i^C(z^{C}) | z_{C_i}=0]\), where the expectation is taken over \(z_{C_{j}} \sim {\rm Bern}(w_j)\) for \(j \neq i\). We taylor expand \(F_i^{1}(\vec{1})\) and \(F_i^{0}(\vec{0})\) around \(F_i^{1}(\vec{p})\) and \(F_i^{0}(\vec{p})\) to the first order and obtain
\begin{align*}
  {\rm ATE} &= F_i^{1}(1) - F_i^{0}(0) \\
  &\approx
  F_i^1({p}) + \sum_{C_j\in \mN_i^C} (1-p)\left({\sf E}[f_i^{C,1}(Z) |  Z_{C_j} = 1] - {\sf E}[f_i^{C,1}(Z) | Z_{C_j}= 0]\right) \\
&\quad  - F_i^0({p}) - \sum_{C_j\in \mN_i^C} (0-p)\left({\sf E}[f_i^{C,0}(Z) | Z_{C_j} = 1] - {\sf E}[f_i^{C,0}(Z) | Z_{C_j} = 0]\right)
\end{align*}
Our estimator Eq.~\ref{eq:DN-clustering} precisely implements these quantities, as analogous to the unit level estimator. It remains to bound the second order error, given as follows.
\begin{align*}
    &|(a-p)^T H_C^a(p')(a-p)| \\
    &= \sum_{C_j,C_k\in \mN^C_i: C_j\neq C_k} (a-p)^2\Bigl(\E[f_i^{C,a}(Z)\mid Z_{C_j} = 1, Z_{C_k} = 1] - \E[f_i^{C,a}(Z)\mid Z_{C_j} = 1, Z_{C_k} = 0]\\
    &\hspace{3cm} - \E[f_i^{C,a}(Z)\mid Z_{C_j} = 0, Z_{C_k} = 1] + \E[f_i^{C,a}(Z)\mid Z_{C_j} = 0, Z_{C_k} = 0]\Bigr)\\
    &= \sum_{C_j,C_k\in \mN^C_i: C_j\neq C_k}(a-p)^2 \E[f^{C,a}_i(Z^{1_{C_j}, 1_{C_k}}) - f_i^{C,a}(Z^{1_{C_j},0_{C_k}}) - f_i^{C,a}(Z^{0_{C_j},1_{C_k}}) + f_i^{C,a}(Z^{0_{C_j},0_{C_k}})]\\
    &=\sum_{\substack{C_j,C_k\in \mN^C_i: C_j\neq C_k\\ |C_j\cap \mN_i| = m, |C_k\cap \mN_i| = n}}(a-p)^2 \E[f^{a}_i(Z^{1_{m}, 1_{n}}) - f_i^{a}(Z^{1_{m},0_{n}}) - f_i^{a}(Z^{0_m,1_{n}}) + f_i^{a}(Z^{0_{m},0_{n}})\mid Z_l = a,\forall l\in C_i]\\
    &\leq \sum_{\substack{C_j,C_k\in \mN^C_i: C_j\neq C_k\\ |C_j\cap \mN_i| = m, |C_k\cap \mN_i| = n}}mn\cdot L\delta^2
\end{align*}
where the second equality is transferring the cluster level $f_i^C$ back to unit level $f_i$ with the conditioning of $z_{C_i} = a$, and the last inequality is due to the previous lemma. Note that here we denote $\mN_i^C$ to be the cluster neighbors of node $i$ without $C_i$, the cluster that $i$ itself resides in. Note that we have $\sum_{C_j\in \mN_i^C} |C_j\cap \mN_i| = d_i - d_i^C$ by definition, where again $d_i^C$ denotes the number of neighbors of node $i$ in cluster $C_i$. Clearly 
$$\sum_{\substack{C_j,C_k\in \mN^C_i: C_j\neq C_k\\ |C_j\cap \mN_i| = m, |C_k\cap \mN_i| = n}}mn\cdot L\delta^2 \leq (\sum_{C_j\in \mN_i^C} |C_j\cap \mN_i|)^2\cdot L\delta^2 = (d_i-d_i^C)^2L\delta^2$$
This completes the proof.

\end{document}